%% file: UCNC.tex
\documentclass[runningheads]{llncs}
\usepackage[T1]{fontenc}
\usepackage{graphicx}
\usepackage[nottoc]{tocbibind}
\usepackage{enumerate}
\usepackage{hyperref}
\usepackage{amssymb}
\usepackage[T1]{fontenc}
\usepackage[utf8]{inputenc}
\usepackage{babel}
\usepackage[font=small,labelfont=bf,tableposition=top]{caption}
\usepackage{booktabs}
\usepackage{longtable}
\usepackage{pdflscape}
\usepackage{pgfplots}
\usepackage{subcaption}
\usepackage{tikz}
\usepackage{float}
\usepackage{microtype}
\usepackage{algorithmic, algorithm}
\usepackage[numbers]{natbib}

\input{math_commands.tex}

\input{proof_macros.tex}
\togglefalse{includeproofs}

\newtoggle{arxiv}
\togglefalse{arxiv}

\begin{document}
\title{Bridging Chaos Game Representations and $k$-mer Frequencies of DNA Sequences}
\titlerunning{Bridging CGR and $k$-mer Frequencies of DNA}
%
\author{Haoze He \inst{1} \and
Lila Kari\inst{2} \and
Pablo Millan~Arias\inst{2}}
\authorrunning{He, Kari, Millan Arias}
%
\institute{\'Ecole Polytechnique F\'ed\'erale de Lausanne (EPFL) \\
\email{haoze.he@epfl.ch}\\ \and
University of Waterloo\\
\email{\{lila,pmillana\}@uwaterloo.ca}}
\maketitle              
\begin{abstract}
\sloppy
This paper establishes formal mathematical foundations linking Chaos Game Representations (CGR) of DNA sequences to their underlying $k$-mer frequencies. We prove that the Frequency CGR (FCGR) of order $k$ is mathematically equivalent to a discretization of CGR at resolution $2^k \times  2^k$,  and its vectorization corresponds to the $k$-mer frequencies of the sequence. Additionally, we characterize how symmetry transformations of CGR images correspond to specific nucleotide permutations in the originating sequences. Leveraging these insights, we introduce an algorithm that generates synthetic DNA sequences from prescribed $k$-mer distributions by constructing Eulerian paths on De Bruijn multigraphs. This 
enables reconstruction of sequences matching target $k$-mer profiles with arbitrarily high precision, facilitating the creation of synthetic CGR images for applications such as data augmentation for machine learning-based taxonomic classification of DNA sequences. Numerical experiments validate the effectiveness of our method across both real genomic data and artificially sampled distributions. To our knowledge, this is the first comprehensive framework that unifies CGR geometry, $k$-mer statistics, and sequence reconstruction, offering new tools for genomic analysis and visualization.  The web application implementing the reconstruction algorithm is available at~ \href{https://tinyurl.com/kmer2cgr}{\url{https://tinyurl.com/kmer2cgr}. 
}
\end{abstract}

\input{Sections/Introduction}

\input{Sections/CGR}
\input{Sections/FCGR}

\input{Sections/k-mers}

\input{Sections/Conclusion}

\begin{credits}
\subsubsection{\ackname} The authors acknowledge the support of the Natural Sciences and Engineering Research Council of Canada (NSERC), [RGPIN-2023-03663] to L.K., and thank Niousha Sadjadi for generating Figure~\ref{fig:cell_structure}.\\

\end{credits}

%
%
%
\bibliographystyle{splncs04}
\bibliography{references}
%





\appendix
\renewcommand{\thesection}{A.\arabic{section}}
\section*{Appendix A: Technical Appendix}
\addcontentsline{toc}{section}{Appendix A: Technical Appendix}
\input{Sections/Technical_Appendix}
\newpage
\renewcommand{\thesection}{B.\arabic{section}}
\section*{Appendix B: Supplementary Information}
\addcontentsline{toc}{section}{Appendix B: General Appendix}
\input{Sections/General_Appendix}

\end{document}

%% file: math_commands.tex
\usepackage{amsmath,amsfonts,bm}

\def\eqref#1{equation~\ref{#1}}

\def\1{\bm{1}}

\def\vtheta{{\bm{\theta}}}

\def\vb{{\bm{b}}}
\def\vc{{\bm{c}}}
\def\vd{{\bm{d}}}

\def\vz{{\bm{z}}}

\DeclareMathAlphabet{\mathsfit}{\encodingdefault}{\sfdefault}{m}{sl}
\SetMathAlphabet{\mathsfit}{bold}{\encodingdefault}{\sfdefault}{bx}{n}



%% file: proof_macros.tex
\usepackage{etoolbox} 
\usepackage{environ}

\newtoggle{includeproofs}

\NewEnviron{conditionalproof}{
  \iftoggle{includeproofs}{
    \begin{proof}
      \BODY
    \end{proof}
  }{}
}

%% file: Sections/Introduction.tex
\section{Introduction}

The increased availability of complete genome sequences has motivated a paradigm shift in comparative genomics, from homology-based to whole-genome analyses based on sequence composition patterns \cite{Zielezinski_AF, cgr_systematic_review_DeLaFuente}. The observation of different structural patterns in DNA sequences dates back to 1990, when  Jeffrey applied concepts from chaotic dynamics to DNA sequences and introduced Chaos Game Representations (CGR)  of DNA sequences \cite{CGR}. A CGR is visualized within a unit square, with each of the four vertices labelled by one of the nucleotides ($A$, $C$, $G$, and $T$). The plotting process follows a simple iterative procedure: the first nucleotide in the sequence is plotted at the midpoint between the center of the square and the vertex corresponding to that nucleotide. Each subsequent nucleotide is then plotted at the midpoint between the previously plotted point and the vertex representing the current nucleotide.

The appearance of interesting geometric patterns in CGRs  of real DNA sequences, such as fractals and parallel lines, motivated further research in the field \cite{Jeffrey1992, dutta1992, Hill1992, goldman1993}. Notably, Oliver \textit{et~al.} \cite{oliver1993} computed the first discretization of CGRs. In these representations, referred to hereafter as Frequency Chaos Game Representations (FCGR), the unit square is divided into a $2^k\times2^k$ grid, and the CGR is discretized by counting the number of points in a given cell. {Each cell is associated with a particular subword of length $k$ ($k$-mer), and in \cite{oliver1993}  it was suggested that the count in each cell must be equal to the frequency of the corresponding $k$-mer in the DNA sequence.
Similarly, Hao introduced $k$-frames \cite{hao2000}, a class of self-similar and self-overlapping fractals mapping each $k$-mer to a cell inside the unit square using the Kronecker product.

\begin{figure}[!tph]
    \centering
    \begin{subfigure}[b]{0.24\linewidth}
        \centering
        \includegraphics[width=\linewidth]{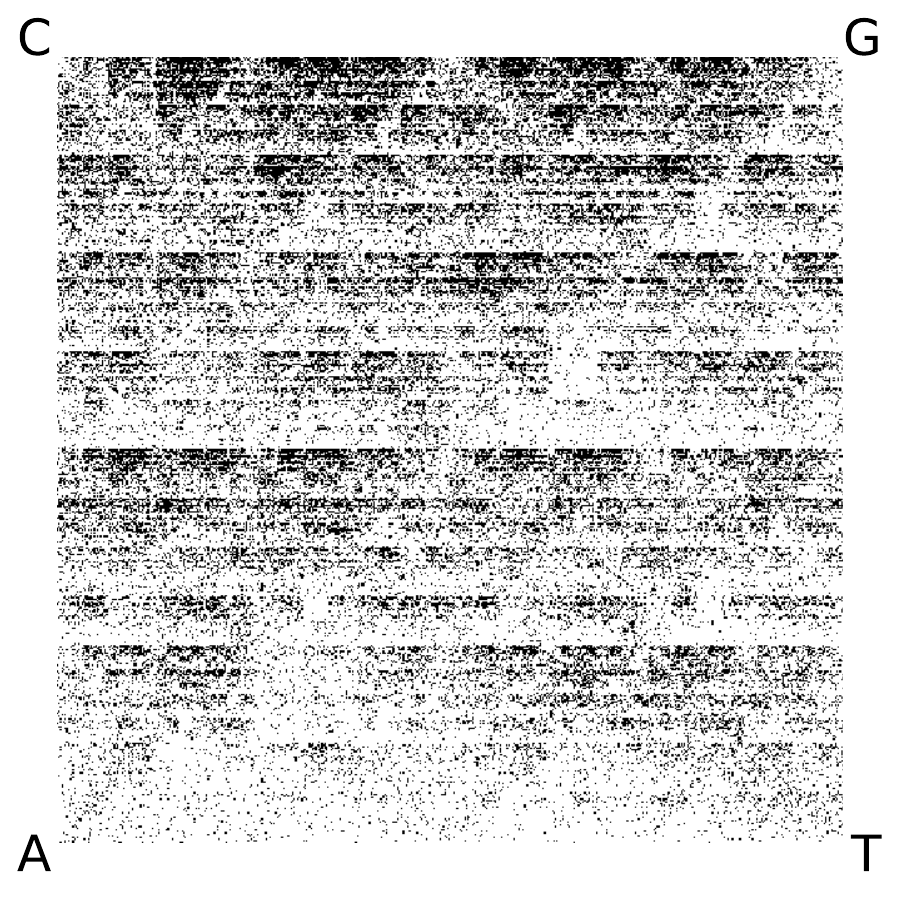}
        \caption{}
        \label{fig:pseudomona_CGR}
    \end{subfigure}
    \hspace{10 pt}
    \begin{subfigure}[b]{0.24\linewidth}
        \centering
        \includegraphics[width=\linewidth]{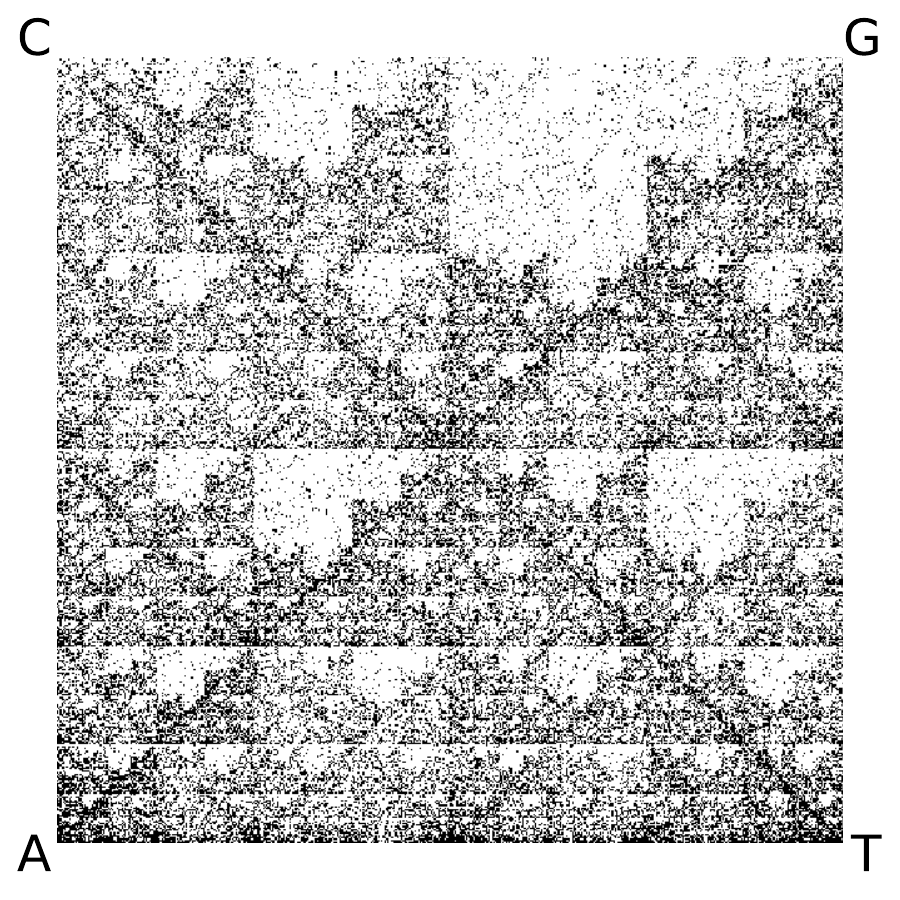}
        \caption{}
        \label{fig:human_cgr}
    \end{subfigure}
    \hspace{10 pt}
    \begin{subfigure}[b]{0.24\linewidth}
        \centering
        \includegraphics[width=\linewidth]{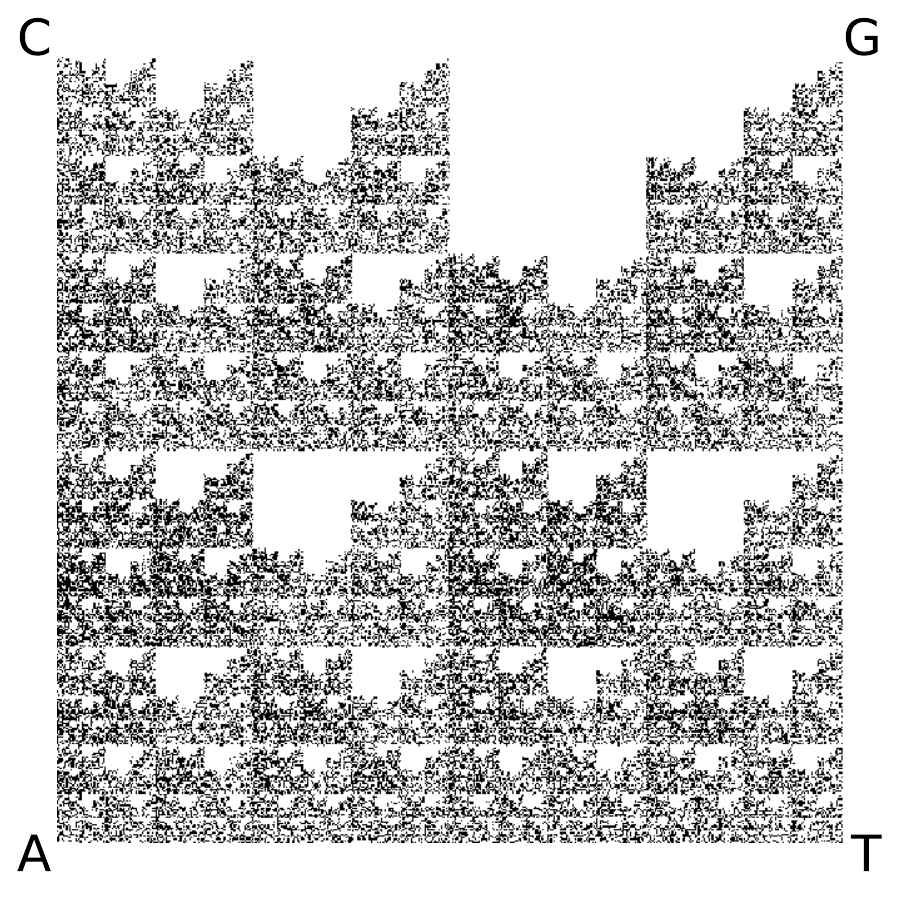}
        \caption{}
        \label{fig:GC_underrepresentation}
    \end{subfigure}
    \caption{(a) CGR of a 100,000 bp  DNA sequence randomly extracted from the complete genome of \textit{Pseudomonas aeruginosa} strain PAO1 (RefSeq NC\_002516.2);  (b) CGR of a 100,000 bp DNA sequence randomly selected from human chromosome 4 (GRCh38.p14 primary assembly, RefSeq NC\_000004.12);  (c) CGR of a computer-generated DNA sequence that is random in all other aspects except that the dinucleotide GC is absent.}
\vspace{-5mm}
\end{figure}

Independent of CGR research, in 1995, Karlin and Burge introduced the notion of a {\it genomic signature} \cite{signature_pops}, as an umbrella term for any numerical quantity that shows greater similarity among DNA sequences of closely related organisms compared to those of more distantly related organisms. In that work, dinucleotide relative abundance profiles (DRAP)  were proposed as genomic signatures, as they were effective in capturing expected variations between some species and similarities within the genome of a single species.   Building on this concept, in 1999, Deschavanne \textit{et~al.} \cite{cgr-deschavanne}, characterized both CGR and FCGR as genomic signatures and showed that the variation between FCGR images of sequences along a genome was smaller than the variation between FCGR images of sequences taken from different genomes. For example, as seen in Figures \ref{fig:pseudomona_CGR} and \ref{fig:human_cgr}, the CGR of a DNA sequence from the genome of \textit{P. aeruginosa} exhibits visual patterns that are significantly different from patterns in the CGR of a DNA sequence from human chromosome 4. Subsequently, these CGR studies lead to the generalization of genomic signatures across various orders $k$, introduced in 2005 by Wang \textit{et~al.}\cite{wang2005} where, e.g.,  $k=2$ corresponds to DRAP. In \cite{wang2005}, it was shown that higher-order FCGRs capture sequence features that are not encoded by the DRAPs that had initially been proposed as genomic signatures. 

The concepts of CGR, FCGR of order $k$, and $k$-mer frequency vector of a sequence (the latter comprising the counts of all its $k$-mers)  collectively referred to here as \textit{genomic signature}, have been widely used in comparative genomics as an alternative or a complement to alignment-based methods. For example, CGRs have inspired numerous alignment-free methods for taxonomic classification \cite{Randhawa_taxonomic_classification_covid19, cgr-kari-karamicahlis, Joseph2006, LochelHeider2021, Rizzo2016}, clustering \cite{millan_arias_delucs_2022, Alipour2024}, and phylogenetic analyses \cite{Hoang2016, Sengupta2020, Lichtblau2019}. Similarly, distance measures between $k$-mer frequency profiles serve as robust proxies for evolutionary relatedness, making them the core component in similar applications \cite{Zielezinski_AF, Solis-Reyes-kameris, AFBenchmark, CART}.

Despite the interchangeable use of these representations of DNA sequences in many applications \cite{AvilaCartes2023, millan_arias_delucs_2022, cgr_systematic_review_DeLaFuente}, a formalization of their interconnections is still missing. In this work, we introduce a theoretical and algorithmic framework that rigorously links these genomic signatures.  The main contributions of this paper are that it: \textit{(i)}  formally establishes two-way connections between a symmetry transformation of the CGR of a sequence and a morphism applied to the underlying sequence (Section \ref{sec:cgr}), \textit{(ii)}  demonstrates the mathematical equivalence between a  CGR  of resolution $2^k \times 2^k$ of a DNA sequence $s$, and the Frequency Chaos Game Representation (FCGR)  of order $k$ of that sequence, as well as $k$-mer frequency vector of the sequence (Section \ref{sec:fcgr}), \textit{(iii)}  provides an algorithm and a software tool that computes a  CGR and the corresponding synthetic DNA sequence,  from a target $k$-mer frequency vector (Section \ref{sec:k_mer_freq}).
\vspace{-3mm}



\subsection{Notation}
Throughout this paper, $\Sigma$ will denote the DNA alphabet, namely the set $\Sigma = \{A, C, G, T\}$.   The cardinality of a set $A$ will be denoted by card$(A)$. A non-empty word (string) over $\Sigma$, $w = a_1 a_2 ... a_n$, $n \geq 1$,  is a concatenation of letters $a_i\in \Sigma$, $1\leq i \leq n$. The empty word is denoted by  $\lambda$, by  $\Sigma^+$ we denote the set of all non-empty words over $\Sigma$,  and $\Sigma^* = \Sigma^+ \cup \{\lambda\}$, while for $k \geq 1$ we have that $\Sigma^k$ is the set of words of length $k$ over $\Sigma$. The length of a word $w$ is denoted by $|w|$, and $|\lambda| = 0$. For a given $k \geq 1$, a word $w \in \Sigma^k$ will be called a $k$-mer.  The number of occurrences of the $k$-mer $w$ in the DNA sequence $s$ will be denoted by $occ(s,w)$.  
 For a word $w \in \Sigma^+$, we will denote by $\operatorname{sub}_k(w)$ the set of all subwords of length $k$ in the word $w$.  We use $\Delta^{n-1}$ to denote the standard simplex $\{(x_1,\ldots,x_n): x_i \geq 0 \text{ for all } i=1,\ldots,n, \sum_{i=1}^{n}x_i = 1\}$ in $\mathbb{R}^{n}$.\\

%% file: Sections/CGR.tex
\section{CGR Symmetries and DNA Letter Permutations}\label{sec:cgr}


This section first recalls the formal definition of the Chaos Game Representation (CGR) of a DNA sequence \cite{CGR} (Definition \ref{def:CGR}) and introduces several other definitions and notations.  Subsequently, Theorem \ref{theorem1} shows that, if a CGR is the image of another CGR via a symmetry transformation of the square, then their originating DNA sequences are connected by a certain letter permutation, and viceversa.   In addition, Theorem \ref{theorem2} establishes the correspondence between the $k$-mers avoided in each of the sequences underlying two CGRs that are obtained from one another by a symmetry transformation of the square.

\begin{definition}[CGR square]\label{def:cgr_sqr}
    A  CGR square is the square centred at the origin 
    \[\{(x,y)|\; -1<x<1, -1<y<1\}\]
    with corners $(-1,-1), (-1,1), (1,1), (1,-1)$, each labeled by the labeling function $label: \Sigma \longrightarrow \mathbb{Z}^2$  
\[
label(A) = (-1,-1),  label(C) = (-1,1), label(G) = (1, 1),  label(T) = (1,-1)
\]
\end{definition}

Note that the CGR square defined above has a side length of 2, whereas the original CGR square defined in~\cite{Jeffrey1992} is a unit square with side length 1.

\begin{definition}[CGR representation of a DNA sequence]\label{def:CGR} 
Let  $n \geq 1$ and let $s  = a_1 a_2 ... a_n$ be a sequence of length $n$ over the DNA alphabet $\Sigma$. The CGR representation  of the sequence $s$ is the  set of points $CGR(s) = \{p_0, p_1, ..,p_n\} \subseteq \mathbb{Q}^2$ whose  coordinates  are defined recursively  by  
$$p_0 = (x_0, y_0)=(0,0), \mbox{ and }  
p_i = \frac{p_{i-1} + label(a_i)}{2} \mbox{ for all }  1 \leq i \leq n.$$ 
\end{definition}

The {\it dihedral group} \cite{Abstractalgebra} of degree 4 and order 8 is the symmetry group of a square
$D_{8}=\left\{e, r, r^{2}, r^{3}, s, s r, s r^{2}, s r^{3}\right\}$, and it comprises rotations,  reflections across the horizontal and vertical axis,  as well as reflections across the diagonals. The symmetries of the  axis-aligned CGR square, centered at the origin, can be represented by  $2 \times 2$  permutation matrices, acting on the plane by multiplication on column vectors of coordinates  $\begin{bmatrix} x \\y
\end{bmatrix}$. The group composition operation is represented as matrix multiplication. \iftoggle{arxiv}{(see Supplementary Section~\ref{subsec:dihedral_group} for details)}{}.

\begin{definition}
 For two words $u, w \in \Sigma^n$ and a symmetry $h \in D_8$,  we will say that $CGR(u) =  h \cdot CGR(w)$ if and only if,  for all 
 $1\leq i \leq n$, we have that
$ q_{u, i} = h \cdot p_{w, i}$
where $q_{u, i}$ is the $i$th point in the generation of $CGR(u)$ and $p_{w, i}$ is the $i$th point in the generation of $CGR(w)$. 
\end{definition}


\begin{definition} Let $CGR(w)$ be the CGR of a  DNA sequence $w \in \Sigma^{*}$. The image  of  a $CGR(w)$ via a transformation  $h \in D_{8}$ is defined as $\left\{h \cdot x|\;  x \in CGR(w)\right\}$, the set obtained by multiplying the matrix $h$  with each point in $CGR(w)$.

\end{definition}

Consider a permutation $\sigma: \Sigma \rightarrow \Sigma$  of the letters in the DNA alphabet $\Sigma$, and consider that $\sigma(\lambda) = \lambda$, where $\lambda$ denotes the empty word. Such a permutation $\sigma$ can be extended to a morphism $\sigma: \Sigma^* \rightarrow \Sigma^*$  by the morphism property whereby   $\sigma(uv) = \sigma(u) \sigma(v)$, for all $u, v \in \Sigma^*$. 
We now explore the relationships between  such morphisms applied to a  DNA sequence and  symmetry transformations in $D_8$ applied to its CGR.

\begin{definition}
\label{defi_morph}
Let $S$ be the following set of permutations extended to morphisms  on $\Sigma^*$:
$$ S = \{(),(A, T, G, C),(A, G)(C, T),$$
$$(A, C, G, T),(A, C)(G, T),(C, T),(A, T)(C, G),(A, G)\}$$ 
\end{definition}

In Definition \ref{defi_morph}, in the customary cyclic notation, $()$ denotes the identity permutation on $\{A, C, G, T\}$,  and, e.g., $(A, T, G, C)$ denotes a circular permutation that maps $A$ to $T$, $T$ to $G$, $G$ to $C$, and $C$ to $A$.  Note also that $S$ is a subset of the $S_4$, set of all 24 permutations of $\{A, C, G, T\}$.

Recall that, \cite[Section 1.2]{Abstractalgebra}, the dihedral group $D_8$ is connected to a subgroup of the permutation group  $S_4$ of all permutations of four elements,  by the mapping $f:D_{8} \rightarrow S_4$ defined as
$f(e) =   (),$ 
$f(r) =   (A, T, G, C),
f(r^{2}) =   (A, G)(C, T),\;
f(r^{3}) =  (A, C, G, T),\;
f(s)=   (A, C)(G, T),$ $
f(s r)=   (A, G), \;
f(s r^{2})=   (A, T)(C, G),\;
f(s r^{3})=   (C, T)$.

The following theorem proves that, if a CGR is observed to be the image of another CGR via one of the symmetry transformations in $D_8$, then their originating DNA sequences are connected by one of the letter permutations in $S$ (see Definition \ref{defi_morph}).

\input{Proofs/Section2/theorem1}

 In other words,  Theorem \ref{theorem1} states that,  given  two DNA sequences $u$ and $w$ and a letter permutation  $\sigma \in S$,  then  $u = \sigma(w)$ iff  $CGR(u)$  can be obtained from  $CGR(w)$  via the symmetry  $f^{-1}(\sigma)$ in $D_8$.  A direct consequence of this result is that the image of a CGR of a DNA sequence under some symmetry transformation $h\in D_8$ is still a CGR of another DNA sequence, i.e., the set of CGRs of DNA sequences is closed under transformations in $D_8$.

\input{Proofs/Section2/cor_closure_of_CGR}



 A DNA sequence $s$  is said to  {\it avoid a word $w$} if $s$ does not contain any occurrence of $w$  as a subword.
 We now investigate the mathematical connection between $k$-mer avoidance in DNA sequences and the family of CGR images related by the symmetry group $D_8$. This question is of interest  because the CGR of a DNA sequence  that avoids  a $k$-mer $w$  has  a similar visual appearance  as the CGR of a real DNA sequence in which  the word $w$ is under-represented, but not completely absent. For instance, Figures~\ref{fig:human_cgr} and~\ref{fig:GC_underrepresentation} illustrate that avoiding the dinucleotide $GC$ produces the same ``double-scoop''  CGR pattern  as that observed in CGRs of human DNA sequences (characterized by an under-representation of $GC$). 
 Given the CGR of a DNA sequence that avoids a certain  $k$-mer, and another CGR obtained from the first through a symmetry transformation in $D_8$, the following theorem identifies the 
$k$-mer avoided by the underlying sequence of the second CGR.

\input{Proofs/Section2/theorem2}





%% file: Proofs/Section2/theorem1.tex
\begin{theorem}
\label{theorem1}
Let $u, w \in \Sigma^n$  be two DNA sequences of length $n$, and let  $\sigma \in S$ be one of the morphisms in $S$.  Then  $u = \sigma(w)$ if and only if  $CGR(u) = f^{-1}(\sigma)\cdot CGR(w)$.


\end{theorem}

\begin{conditionalproof}
Recall the labelling function $label$ that associates each letter of the DNA alphabet $\Sigma$  to one of the corners of the CGR square.
We start by proving the following claim.\\

 \noindent
 {\bf Claim 1.}  Let $\sigma \in S$ be one of the letter permutations on $\Sigma$,   extended to a morphism on $\Sigma^*$, and let $f: D_8 \rightarrow  S$ be the mapping that associates to every symmetry of the square in $D_8$ one such letter permutation in $S$.  Then, for any 
letters $ a, b \in \Sigma$, we have that 
 $$  a = \sigma(b) \iff  label(a) = f^{-1}(\sigma)\cdot label(b).$$ 

\noindent
{\it Proof of Claim 1.}
Assume, without loss of generality, that $b = G$. This implies $label(b) = (1,1)$. The following cases are possible  for the morphism $\sigma$:

\begin{enumerate}
    \item  $\sigma = ()$, the identity permutation on $\Sigma$.
    
    ``$\Longrightarrow$"   $a = \sigma(b) = \sigma(G) = G$  implies   $label(a) = label(G) =  (1, 1)$ and $ f^{-1}(\sigma) = e$. Thus, $  f^{-1}(\sigma)  \cdot  label(b) = e \cdot label(G)  =  label(G) = label(a)$. 
  
    ``$\Longleftarrow$" $f^{-1}(\sigma) = \begin{bmatrix}
        1 & 0 \\ 0 & 1
    \end{bmatrix},$  which implies $label(a) =f^{-1}(\sigma)\cdot label(b)  = (1,1)$, further implying that $a = G = \sigma(b)$. 
    
    \item  $\sigma = (A\ C\ G\ T)$.
    
     ``$\Longrightarrow$"  $a = \sigma(G) = T$ implies  $label(a) = label(T) = (1, -1)$ and $f^{-1}(\sigma) =  r^3$. Thus, \[ f^{-1}(\sigma)  \cdot label(b)  = r^3 \cdot (1, 1) =   \begin{bmatrix} 0 & 1 \\ -1 & 0  \end{bmatrix} \cdot (1, 1) = (1, -1) = label(a).\]
   ``$\Longleftarrow$"   $f^{-1}(\sigma) = \begin{bmatrix}
        0 & 1 \\ -1 & 0
    \end{bmatrix},$ which implies  $label(a) =f^{-1}(\sigma)\cdot label(b)  = (1,-1)$, which further  implies that $a = T = \sigma(b)$. 
    
    \item $\sigma = (A\ G)(C\ T)$.
    
   ``$\Longrightarrow$"   $a = \sigma(G) =A$ implies $label(a) =  label(A) = (-1,-1)$ and  $f^{-1}(\sigma) = r^2$. Thus, 
    \[ f^{-1}(\sigma)  \cdot label(b)  = r^2 \cdot (1, 1)   = \begin{bmatrix} -1 & 0 \\ 0 & -1 \end{bmatrix}  \cdot (1, 1) =(-1, -1) = label(a).\]
    
      ``$\Longleftarrow$"   $f^{-1}(\sigma) = \begin{bmatrix}
        -1 & 0 \\ 0 & -1
    \end{bmatrix}$ implies $label(a) =f^{-1}(\sigma)\cdot label(b)  = (-1,-1)$, which further implies that $a = A = \sigma(b)$. 
    
    \item  $\sigma = (A\ T\ G\ C)$.
    
    ``$\Longrightarrow$"   $a = \sigma(G) =C$ implies $label(a) =  label(C) = (-1,1)$ and $f^{-1}(\sigma) = r$. Thus, 
    \[ f^{-1}(\sigma)  \cdot label(b)  = r \cdot (1, 1)   = \begin{bmatrix} 0 & -1 \\ 1 & 0 \end{bmatrix} \cdot (1, 1) = (-1, 1) = label(a).\]
        ``$\Longleftarrow$"   $f^{-1}(\sigma) = \begin{bmatrix}
        0 & -1 \\ 1 & 0
    \end{bmatrix}$ implies $l(a) =f^{-1}(\sigma)\cdot label(b)  = (-1,1)$, which further implies that $a = C = \sigma(b)$.

    \item  $\sigma = (A\ C)(G \ T) $.
    
    ``$\Longrightarrow$" $a = \sigma(G) =T$ implies $label(a) =  label(T) = (1,-1)$ and  $f^{-1}(\sigma) =s$. Thus, 
    \[ f^{-1}(\sigma)  \cdot label(b)  = s \cdot (1, 1)   = \begin{bmatrix} 1 & 0 \\ 0 & -1 \end{bmatrix} \cdot (1, 1) = (1, -1) = label(a).\]
     ``$\Longleftarrow$"  $f^{-1}(\sigma) = \begin{bmatrix}
        1 & 0 \\ 0 & -1
    \end{bmatrix}$ implies $label(a) =f^{-1}(\sigma)\cdot label(b)  = (1,-1)$, which further implies that $a = T = \sigma(b)$. 
    
    \item  $\sigma = (A\ T)(C \ G) $.
    
    ``$\Longrightarrow$"  $a = \sigma(G) =C$ implies  $label(a)  =  label(C) =(-1,1)$ and $f^{-1}(\sigma) = sr^2$. Thus, 
    \[ f^{-1}(\sigma)  \cdot label(b) = sr^2 \cdot (1, 1) =   \begin{bmatrix} -1 & 0 \\ 0 & 1 \end{bmatrix} \cdot (1, 1) = (-1, 1) = label(a).\]
    ``$\Longleftarrow$"  $f^{-1}(\sigma) = \begin{bmatrix}
        -1 & 0 \\ 0 & 1
    \end{bmatrix}$  implies  $label(a) =f^{-1}(\sigma)\cdot label(b)  = (-1,1)$, which further  implies that $a = C = \sigma(b)$. 
    
    \item  $\sigma = (C\ T)$.
    
    ``$\Longrightarrow$"  $a = \sigma(G) = G$ implies  $label(a)  = label(G)=(1,1)$ and $f^{-1}(\sigma)  = sr^3$. Thus, 
 \[  f^{-1}(\sigma)  \cdot label(b)=  sr^3 \cdot (1, 1)  = \begin{bmatrix} 0 & 1 \\ 1 & 0 \end{bmatrix} \cdot (1,1)= (1,1) = label(a).\]
    ``$\Longleftarrow$"  $f^{-1}(\sigma) = \begin{bmatrix}
        0 & 1 \\ 1 & 0
    \end{bmatrix}$  implies $label(a) =f^{-1}(\sigma)\cdot label(b)  = (1, 1)$, which further implies that $a = G = \sigma(b)$. 
    
    \item $\sigma = (A\ G)$.
    
    ``$\Longrightarrow$" $a = \sigma(G) = A$ implies $label(a) = label(A) = (-1,-1)$ and  $f^{-1}(\sigma) =  sr $. Thus, 
    \[ f^{-1}(\sigma)  \cdot label(b) = sr \cdot (1, 1)  = \begin{bmatrix}
        0 & -1 \\ -1& 0
    \end{bmatrix}(1,1)= (-1,-1) = label(a).\]
      ``$\Longleftarrow$"  $f^{-1}(\sigma) = \begin{bmatrix}
        0 & -1 \\ -1 & 0
    \end{bmatrix}$ implies $label(a) =f^{-1}(\sigma)\cdot label(b)  = (-1,-1)$, which further implies that $a = A = \sigma(b)$. 

\end{enumerate}

\noindent
The cases when $b = T$, $b = A$, and $b = C$ can be proved similarly, and complete proof of   Claim 1.\\

\noindent
The proof is now by induction on $n$, the length of the input strings $u$ and $w$. Assume $w = b_1 b_2\ldots b_n$ and $u = a_1 a_2\ldots a_n$, where $n \geq 1$ and  $a_i, b_i$ are letters in $\Sigma$ for all $1\leq i\leq n$.  Recall that $CGR(u) = f^{-1}(\sigma) \cdot CGR(w)$ iff,    for all
$1 \leq i \leq n$, we have 
$q_{i} = f^{-1}(\sigma)\cdot p_{i}$
where  $q_{i}$ is the $i$-th point in the generation  of  $CGR(u)$, and   $p_{i}$ is the $i$-th point  in the generation of $CGR(w)$.\\

\noindent
{\bf Base Case.}  $n = 1$.  Then $w = b \in \Sigma$, and $u = a \in \Sigma$.

``$\Longrightarrow$"
   Assume $a = \sigma(b)$. By the recursive definition of CGR, $CGR(a)$ consists of a single point $q_1 = \frac{label(a)}{2}$, and $CGR(b)$ consists of a single point  $p_1 = \frac{label(b)}{2}$.  By Claim 1,  $label(a) = f^{-1}(\sigma) \cdot label(b)$.  Thus, 
   $$q_1 = \frac{label(a)}{2} =  \frac{f^{-1}(\sigma)\cdot label(b)}{2} =  f^{-1}(\sigma)\cdot  p_1$$

   ``$\Longleftarrow$"   Assume $CGR(a) = f^{-1} (\sigma) \cdot CGR(b)$. Since $CGR(a)$  consists of one point  $q_1 = \frac{label(a)}{2}$, and $CGR(b)$ consists of one point  $p_1 = \frac{label(b)}{2}$, this implies 
   $$q_1 = \frac{label(a)}{2} = f^{-1}(\sigma) \cdot p_1  = f^{-1}(\sigma) \cdot  \frac{label(b)}{2}, $$
   further implying  $label(a) = f^{-1}(\sigma) \cdot label(b)$. By Claim 1 this now implies that $a = \sigma(b)$.\\

\noindent
{\bf Inductive step.}  Assume that the iff statement of Theorem \ref{theorem1}  holds for  words of length $(m-1)$,  with $m \geq 2$ (Inductive Hypothesis, I.H.). Let $w$ and $u$ be two words of length $m$, that is, $w = b_1 b_2\ldots b_{m-1} b_m$ and  $u = a_1 a_2\ldots a_{m-1} a_m$, where $a_i, b_i \in \Delta$ for all $1\leq i\leq m$.

``$\Longrightarrow$"  Assume $u = \sigma(w)$. By the definition of morphism, we have that $u = \sigma(w)$ iff $a_i = \sigma(b_i)$, for all $1\leq i \leq m$.

\noindent
Denote 
\[w' = b_1 b_2\cdots b_{m-1}, \quad u' = a_1 a_2\cdots a_{m-1}.\]

\noindent
By (I.H.),  since  $u' = \sigma(w')$, we have that $q_{i} =  f^{-1}(\sigma) \cdot  p_{i}$, for all $1\leq i \leq m-1$.

\noindent
By the definition of CGR, the $m$th point of  $CGR(u)$ is $q_m$, and the $m$th point of $CGR(w)$ is $p_m$, where
\[q_m = \frac{q_{m-1} + label(a_m)}{2}, \mbox{ and  } p_m =  \frac{p_{m-1} + label(b_m)}{2}.\]

\noindent
Since $a_m = \sigma(b_m)$, by Claim 1 we have that 
$label(a_m) =f^{-1}(\sigma) \cdot label(b_m)$.  Thus, by I.H. and Claim 1, we have 
\[q_m = \frac{q_{m-1} + label(a_m)}{2} = 
\frac{f^{-1}(\sigma) \cdot p_{m-1} + f^{-1}(\sigma) \cdot label(b_m)}{2} = \]
\[f^{-1}(\sigma) \cdot  \frac{p_{m-1} + label(b_m)}{2} = f^{-1}(\sigma) \cdot p_m.\]

\noindent
Since the equality  $q_{i} =  f^{-1}(\sigma) \cdot  p_{i}$  holds  now for all $1\leq i \leq m$, it follows that $CGR(u) = f^{-1}(\sigma) \cdot CGR(w)$.\\



 
``$\Longleftarrow$" Assume that   $CGR(u) = f^{-1}(\sigma) \cdot CGR(w)$. This implies that, for all $1\leq i \leq m$ we have that $q_i = f^{-1}(\sigma)  \cdot p_i$. \\

\noindent
The recursive definition of $CGR(w)$ and the previous observation imply that: 
\begin{equation}
\label{eq:u_component}
q_m = f^{-1}(\sigma)\cdot p_m = \frac{q_{m-1} + label(a_m)}{2},
\end{equation}
In parallel, the recursive definition of CGR and the fact that $f^{-1}(\sigma)$ is a transformation in $\mathbb{R}^2$ imply that: 
\begin{equation}
\label{eq:w_component}
f^{-1}(\sigma) \cdot p_m = f^{-1}(\sigma)  \cdot \left ( \frac{p_{m-1} + label(b_m)}{2} \right ) 
\end{equation}
By transitivity of equality applied to equations \ref{eq:u_component} and \ref{eq:w_component} we get that:
$$\frac{q_{m-1} + label(a_m)}{2} = f^{-1}(\sigma) \cdot \left ( \frac{p_{m-1} + label(b_m)}{2} \right )$$
This can be further simplified using the distributivity property of matrix multiplication: 
\begin{align*}
q_{m-1} + label(a_m) &= f^{-1}(\sigma) \cdot p_{m-1} + f^{-1}(\sigma) \cdot label(b_m)\\
label(a_m) &= f^{-1}(\sigma) \cdot label(b_m)
\label{eq:points_2_letters}
\end{align*}
which implies that $a_m = \sigma(b_m)$ by Claim 1. 

By the (I.H.),  since {$CGR(u') = f^{-1}(\sigma)\cdot CGR(w')$} we have that  $u' = \sigma(w')$. Thus, by the definition of a morphism and the previously derived relation $a_m = \sigma(b_m)$, we obtain that $\sigma(w) = \sigma(w')\sigma(b_m) = u'a_m = u$.
\end{conditionalproof}

%% file: Proofs/Section2/cor_closure_of_CGR.tex
\begin{corollary}
\label{cor:closure_of_GR}
    Given a word   $u\in\Sigma^n$ and   a symmetry transformation $h \in D_8$,  there exists a word $w \in \Sigma^n$    such that  $CGR(w) = h \cdot CGR(u)$. Constructively, this word can be computed as $w = f(h)\cdot u$, with $f: D_8 \to S_4$  as defined above. 
\end{corollary}
\begin{conditionalproof}
    It directly follows from Theorem~\ref{theorem1}.
\end{conditionalproof}

%% file: Proofs/Section2/theorem2.tex
\begin{theorem}
\label{theorem2}
Let $w, u \in \Sigma^n$ be two DNA words of length $n$,  and let $\alpha, \beta \in \Sigma^k$ be two $k$-mers, $k < n$,  such that $sub_k(w) = \Sigma^k \backslash\{ \alpha \}$ and $sub_k(u) = \Sigma^k \backslash\{ \beta \}$. Then, for any permutation morphism $\sigma\in S$, we have that $CGR(u) = f^{-1}(\sigma) \cdot   CGR(w)$ implies  $\beta = \sigma(\alpha)$. 
\end{theorem}

\begin{conditionalproof}
Let $\sigma \in S$ be one of the permutation morphisms in $S$.\\

\noindent
``$\Longrightarrow$''  Since $CGR(u) = f^{-1}(\sigma)  \cdot CGR(w)$, by Theorem \ref{theorem1}  we have $u =\sigma(w)$.\\ 

\noindent
 Since $\sigma$ is a morphism, it follows that $sub_{k}(u) = sub_{k}(f(w)) = f(sub_k(w))$.

\noindent
Now let us show $\beta = \sigma(\alpha)$.  We have that
$$sub_{k}(u) = \Sigma^k \setminus 
\{\beta \} = sub_{k}(\sigma(w)) 
= \sigma(sub_k(w)) = \sigma(\Sigma^k \setminus \{\alpha \})$$

\noindent
Since $\sigma \in S$ is a permutation of the letters of $\Sigma$, we have that  $\sigma(\Sigma^k) = \Sigma^k$, and   $\sigma(\Sigma^k \setminus \{\alpha \}) = \Sigma^k \setminus \{\sigma(\alpha) \}$. Together with the above equalities, this implies $\beta = \sigma(\alpha)$.\\ 


\end{conditionalproof}

%% file: Sections/FCGR.tex
\section{\label{sec:fcgr}
Frequency CGR, and  Equivalence to  $k$-mer Frequencies}


In this section, we first recall the definition of $FCGR_k(s)$ \cite{cgr-deschavanne}, the frequency CGR of order $k$, of a sequence $s$,  which is discretized version of  $CGR(s)$  at resolution $2^k \times 2^k$. We formalize the discretization process of $CGR(s)$  (Proposition \ref{prop:correspondence}), and establish several properties of 
 this discretized version of $CGR(s)$. Lastly, we prove that $FCGR_k(s)$  can be equivalently computed either by discretizing $CGR(s)$ at resolution $2^k \times 2^k $, or by directly counting the occurrences of all $k$-mers in the sequence $s$ (Theorem \ref{thm:main_thm_fcgr} and its Corollary \ref{corollary:cgr2fgr}). 

We start by describing the discretizing process of the CGR of a DNA sequence, by subdividing  the CGR square (of size $2 \times 2$)  into  $2^k \times 2^k$ equal sub-squares of size $1/2^{k-1}$.

\begin{definition}[grid cell of order $k$]\label{def:cell_res_k}
    Given a CGR square,  $k\geq 1$,  and indices  $0 \leq i, j \leq  2^k -1$,  the {\it grid cell $(i, j)$ of order $k$}, denoted by  $\text{cell}_k(i,j)$,  is the region of the square  defined by
\[cell_k(i,j) = \{(x,y):  x_i - \frac{1}{2^k}< x < x_i + \frac{1}{2^k}, y_j -\frac{1}{2^k} <y < y_j + \frac{1}{2^k}\} \]
where  the center of the grid cell $cell_k(i,j)$  is the point
    \begin{equation}\label{eq:center_subsquare}
    (x_j,y_i) = \left(-\frac{2^k-1}{2^k} + \frac{j}{2^{k-1}}, \frac{2^k-1}{2^k} - \frac{i}{2^{k-1}}\right).
\end{equation}
\end{definition}

    Note that the size of each grid cell $cell_k(i,j)$ of order  $k$ is $\frac{1}{2^{k-1}}  \times \frac{1}{2^{k-1}}$.  Observe also that each grid cell $cell_k(i,j)$ (without boundary) $i = 0,\ldots,2^k-1, j= 0,\ldots,2^k-1$, is uniquely determined by its center $(x_j,y_i)$, where the indices $i$ and $j$ are swapped so as to match the indexing convention of matrices.  
    Lastly, note that grid cells of order $k$ correspond to an image resolution of $2^k \times 2^k$.
    
     We can now define the Frequency Chaos Game Representation of order $k$  of a DNA sequence.
\begin{definition}[FCGR of order $k$, of a sequence $s$]\label{def:grid_def}
Let $k\geq 1$ and $n\ge  k$. A Frequency Chaos Game Representation of order $k$, of a sequence $s \in \Sigma^n$,  is a matrix $FCGR_k(s) \in \mathbb{N}^{2^k \times 2^k}$  whose entries are defined as 
$$FCGR_k(s)(i, j) = \text{card}\left(CGR(s) \cap cell_k(i, j)\right) \text{ for all } 0\leq i, j\leq 2^k-1.$$
\end{definition}

The $FCGR_k(s)$ matrix provides a compressed representation of the sequence $s$. Definition \ref{def:grid_def} can be viewed as a discretization of  $CGR(s)$ into $2^k \times 2^k$ cells, each counting the number of points of $CGR(s)$  that fall inside that cell. We will show later that counting the points in a grid cell of order $k$  is equivalent to calculating the number of occurrences of a specific $k$-mer in the sequence $s$. For that, we need to define the notion of a {\it CGR cell associated with a $k$-mer}.

\begin{definition}\label{def:last_cgr_point}
Let $n \ge 1$ and $s \in \Sigma^n$.  The last point of  the CGR representation of $s$, $CGR(s) = \{p_0, p_1,\ldots,p_n\}$, is defined as  $p_{last}(s) = p_n$.
\end{definition}


 Let  $p_{last}(w) = (x_w, y_w) $ be the last point in  $CGR(w)$ of a word  $w=a_1 a_2 ... a_k$, where $k \geq 1$,  and assume that $label(a_l) = (x_l, y_l)$ for all $1\leq l \leq k$.  By the definition of $CGR(s)$, it is easy to see that 
 \begin{equation}\label{eq:P_last(w)}
   x_w = \frac{\Sigma_{l=1}^k x_l \cdot 2^{l-1}}{2^k}, \quad y_w = \frac{\Sigma_{l=1}^k y_l \cdot  2^{l-1}}{2^k}. 
 \end{equation}
As observed in \cite{Hoang2016},  distinct sequences $s \in \Sigma^+$ have distinct last  points $p_{last}(s)$ in their respective CGR representations. 

\begin{definition}[CGR cell associated with a $k$-mer]\label{def:cell_of_k_mer}
    Let $k\geq 1$,  let $w$ be a $k$-mer in $\Sigma^k$, and let $p_{last}(w) = (x_w,y_w)$ be the last point of  $CGR(w)$.
    The {\it CGR cell associated with the $k$-mer  $w$} is
    \[c(w) = \{(x,y)|\;\;  x_w - \frac{1}{2^{|w|}}< x < x_w + \frac{1}{2^{|w|}}, y_w -\frac{1}{2^{|w|}} <y < y_w + \frac{1}{2^{|w|}}\}.\]
\end{definition}

 Figure \ref{fig:cell_of_k_mer} illustrates $c(ACG)$,  the CGR  cell associated with the $3$-mer ACG.  It is easy to see that, for a $k$-mer $w$, the last point  $p_{last}(w)$ of $CGR(w)$ is the center of the sub-square $c(w)$. Note also that for two different $k$-mers $w_1$ and $w_2$, the cells  $c(w_1)$  and $c(w_2)$ do not intersect.

     


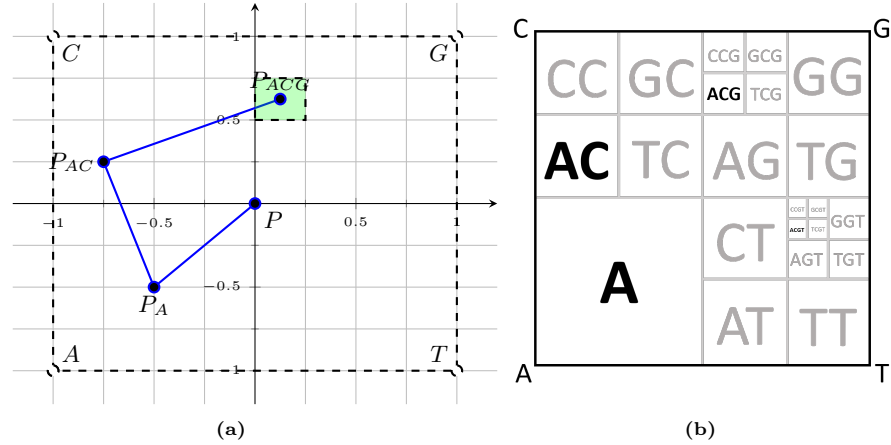
\begin{figure}[!h]
\centering
\hspace*{-10pt}
\begin{subfigure}[b]{0.48\linewidth}
        \begin{tikzpicture}
\begin{axis}[
width=1.37\linewidth,
axis lines=middle,
grid=both,
ymin=-1.2,
ymax=1.2,
yticklabel style={font=\tiny},
minor tick num=1,
xmin=-1.2,
xmax=1.2,
xticklabel={ },
xticklabel style={font=\tiny}
]
\addplot[color=black,dashed,thick,mark=*, mark options={fill=white}] 
    coordinates {
         (-1,-1)
         (-1,1)
         (1,1)
         (1,-1)
         (-1,-1)
        }; 
\node [above right] at (axis cs:  -1, -1) {$A$};
\node [below right] at (axis cs:  -1, 1) {$C$};
\node [below left] at (axis cs:  1, 1) {$G$};
\node [above left] at (axis cs:  1, -1) {$T$};
\addplot[color=blue,solid,thick,mark=*, mark options={fill=black}] 
    coordinates {
         (0,0)
         (-0.5,-0.5)
         (-0.75, 0.25)
         (0.125, 0.625)

        };

\addplot [
        thick,
        dashed,
        color=black,
        fill=green, 
        fill opacity=0.25
    ] 
    coordinates {
         (0,0.5)
         (0.25,0.5)
         (0.25, 0.75)
         (0.0, 0.75)
        (0,0.5)
        };
\node [below right] at (axis cs:  0, 0) {$P$};
\node [below] at (axis cs:  -0.5, -0.5) {$P_A$};
\node [left] at (axis cs:  -0.75, 0.25) {$P_{AC}$};
\node [above] at (axis cs:  0.125, 0.625) {$P_{ACG}$};

\end{axis}
\end{tikzpicture}
        \caption{}
        \label{fig:cell_of_k_mer}
    \end{subfigure}
    \hspace{10pt}
    \begin{subfigure}[b]{0.45\linewidth}
        \centering
        \includegraphics[width=1.0\linewidth]{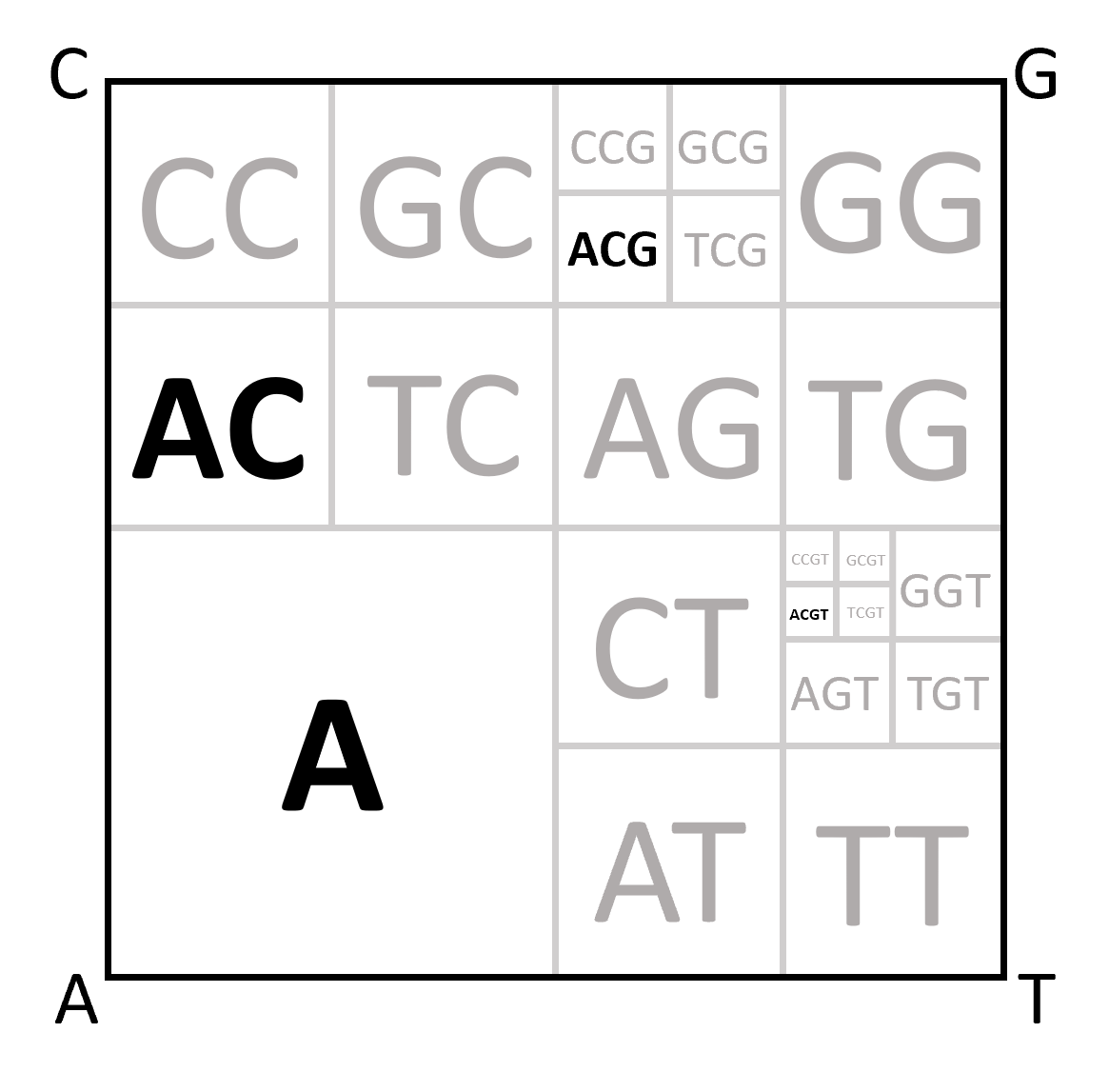}
        \vspace*{-12 pt}
        \caption{}
        \label{fig:cell_structure}
    \end{subfigure}
\caption{
         {\bf (a)} The cell $c(ACG)$ associated with the $3$-mer ACG is marked in green. Its center $P_{ACG}$ coincides with $p_{last}(ACG)$, the last point in the CGR representation of $ACG$. {\bf (b)} Illustration of the hierarchical structure of CGR cells described by Lemma   \ref{lem:inclusion} and Proposition~\ref{prop:cell_hierarchy}  (adapted from~\cite{LochelHeider2021}). }

\vspace{-5mm}
\end{figure}



The following result now establishes the correspondence between the grid cells of order $k$  (Definition~\ref{def:cell_res_k}) and the  CGR cells associated with a $k$-mer  (Definition~\ref{def:cell_of_k_mer}). 

\input{Proofs/Section3/prop_correspondence}

   For a given $k$,  Proposition~\ref{prop:correspondence}  establishes a bijection between the set of  CGR cells associated with  $k$-mers  (Definition~\ref{def:cell_of_k_mer}) and the set of grid cells of order $k$  (Definition~\ref{def:cell_res_k}).
Observe that,  in both Definition \ref{def:cell_res_k} and Definition \ref{def:cell_of_k_mer}, the respective cells do not include their boundaries. The next proposition, aided by the following lemma,  proves that no point of $CGR(s)$ falls on the boundary of a cell associated with a $k$-mer, for any  $k$ with $1 \leq k \leq n$, where $|s| = n$ (ensuring thus that the $FCGR_k(s)$ as defined in Definition~\ref{def:grid_def} does not miscount). 


\input{Proofs/Section3/lem_decom}

\input{Proofs/Section3/prop_boundary}




Lemma \ref{lem:decom} also serves as an aide to showing the hierarchical nested structure of the CGR cells associated with a $k$-mer $w$, as defined in Definition~\ref{def:cell_of_k_mer}. The next result shows that the CGR cell associated with a given $k$-mer $w$ includes all the CGR cells associated with words ending in $w$.

\input{Proofs/Section3/lem_inclusion}
 The {\it closed CGR cell associated with a $k$-mer $w$}  is denoted by 
$\overline{c(w)}$ and is defined as  the union between $c(w)$ and its boundaries (sides of the square).
The next result proves that, given $n \geq 1$,  the closed CGR cell associated with a $k$-mer $w$ equals the union of all closed CGR cells associated with words $w' w$ with $|w'| = n$. This hierarchical structure of the closed cells associated with $k$-mers is illustrated in Figure~\ref{fig:cell_structure}.

\input{Proofs/Section3/prop_cell_hierarchy}
 Lemma \ref{lem:decom} can now be used to show that the $FCGR_k(s)$  matrix defined in Definition~\ref{def:grid_def} through discretizing $CGR(s)$ into grid cells of order $k$ (i.e., at resolution $2^k \times 2^k$), can also be obtained directly by counting the number of occurrences of all $k$-mers in the sequence $s$.

\input{Proofs/Section3/thm_main_thm_fcgr}


\input{Proofs/Section3/corollary_cgr2fgr}
Importantly, we now observe that in  \cite{hao2000} one finds a definition of   ``FCGR of order $k$ of a sequence $s$'' via Kronecker products \iftoggle{arxiv}{(Definition~\ref{def:freq_def} in \cite{hao2000})}{} that is {\it different} from our Definition \ref{def:grid_def}.  The advantage of the concept as defined in Definition \ref{def:grid_def} is that, besides being equivalent to computing the counts  of $k$-mers, $FCGR_k(s)$ in Definition \ref{def:grid_def} is a discretization of  $CGR(s)$, at  resolution $2^k \times 2^k$, and it is thus both visually similar and obtainable from $CGR(s)$. In contrast, the FCGR of order $k$   of a sequence $s$ as defined in \cite{hao2000} is not connected to $CGR(s)$, and it is visually different from $CGR(s)$.
(see Appendix~\ref{subsec:fcgr_kron} for a detailed comparison).

%% file: Proofs/Section3/prop_correspondence.tex
\begin{proposition}\label{prop:correspondence}
    Let $k\geq 1$,  let  $w = a_1\cdots a_k$ be a $k$-mer, and assume that $label(a_l) = (x_l, y_l)$ where $x_l,y_l\in\{-1,1\}$, for all $ 1\leq l \leq k$.     
    Then $c(w) = cell_k(i,j)$, that is,  the CGR cell associated with the $k$-mer $w$ equals  the grid cell $(i, j)$ of order $k$ with indices
     \begin{equation}\label{eq:index}
    j = \frac{2^k-1 + \Sigma_{l=1}^k x_l \cdot 2^{l-1} }{2}, \quad i = \frac{2^k -1 - \Sigma_{l=1}^k y_l\cdot 2^{l-1} }{2}.
     \end{equation}
    
\end{proposition}
\begin{conditionalproof}
We start by first observing that,  for a given value of $k$ and  $k$-mer $w$, the size of the CGR cell $c(w)$  is the same as the size of  $cell_k(i,j)$,  a grid cell of order $k$, where $0\leq i, j \leq 2^{k-1}$. Therefore, for a given $k$ and  $k$-mer $w$, what remains is to find indices  $0\leq i, j \leq 2^{k-1}$  such that the center $(x_j, y_i)$ of the grid cell $cell_k(i,j)$ coincides with the center $p_{last}(w) = (x_w, y_w)$  of $c(w)$, if such indices  exist.

Let $w=a_1a_2\cdots a_k$ and let  $p_{last}(w)$ be the last point in  $CGR(w)$.

Recall that $p_{last}(w) = (x_w, y_w)$ in \eqref{eq:P_last(w)} is the  center of the $c(w)$. For any indexing $(i,j)$, the center $(x_j,y_i)$ of the $cell_k(i, j)$ is defined by \eqref{eq:center_subsquare}. Letting $(x_w,y_w) =(x_j,y_i)$ leads to
    \[x_w =\frac{\Sigma_{l=1}^k x_l\cdot 2^{l-1}}{2^k} = x_j = -\frac{2^k-1}{2^k} + \frac{j}{2^{k-1}}, \]
    \[ y_w = \frac{\Sigma_{l=1}^k y_l\cdot 2^{l-1}}{2^k} =  y_i = \frac{2^k-1}{2^k} - \frac{i}{2^{k-1}}.  \] 
  
  Solving for $j$ and $i$ gives 
 \[
    j = \frac{2^k-1 + \Sigma_{l=1}^k x_l \cdot 2^{l-1} }{2}, \quad i = \frac{2^k -1 - \Sigma_{l=1}^k y_l\cdot 2^{l-1} }{2},
  \]
  which is \eqref{eq:index} as required. Since  $x_l\in \{-1,1\}$ for all $ l=1,\ldots,k$, we have that $j$ is an integer. Furthermore, the minimum  value of $j$ is achieved when $x_l = -1$  for all $ l=1,\ldots,k$, and its maximum is achieved when  $x_l = 1$  for all $ l=1,\ldots,k$, which leads to $j \in \{0,\ldots,2^k - 1\}$. By a similar reasoning, $i \in \{0,\ldots,2^k-1\}$. Thus, for these values of $i$ and $j$, we have that   $(x_w,y_w)=(x_j, y_i)$. Therefore, the grid cell   $cell_k(i, j)$ equals the CGR cell $c(w)$ associated with the $k$-mer $w$,  as having the same size and the same center.
    
   
\end{conditionalproof}

%% file: Proofs/Section3/lem_decom.tex
\begin{lemma}\label{lem:decom}
    Let $n \ge 1,k \ge 1$, and let  $u \in \Sigma^{n}, w \in \Sigma^k$ be two DNA sequences. Let $p_{last}(u),p_{last}(w)$ be the last points of $CGR(u), CGR(w)$ respectively. Then, the last point of $CGR(uw)$ is 
    \[p_{last}(uw)= \frac{p_{last}(u)}{2^{|w|}} + p_{last}(w).\]
\end{lemma}

\begin{conditionalproof}
    By induction on $k$.\\
    \noindent
{\bf Base Case.}  Let $k = 1$, that is, $w = a $, where $a \in \Sigma$. It follows from the definition of CGR that 
\[p_{last}(ua) = \frac{p_{last}(u)}{2} + \frac{label(a)}{2} = \frac{p_{last}(u)}{2} + p_{last}(a).  \]
\\
\noindent
{\bf Inductive step.}  Assume that the statement holds for an arbitrary $k \geq 1$, and let  $w \in \Sigma^{k+1}$. Then $ w = w' a_{k+1}$ where $w' \in \Sigma^k, a_{k+1} \in \Sigma$. By  I.H., for all $u \in \Sigma^{n_1}, n_1 \ge 0$, we have that the last point $p_{last}(uw')$ of $CGR(uw')$ is
\[ p_{last}(uw')  = \frac{p_{last}(u)}{2^k} + p_{last}(w').\]
The last point $p_{last}(uw)$ of $CGR(uw)$ is

\begin{align*}
    p_{last}(uw) = p_{last}(u w' a_{k+1}) &= \frac{p_{last}(uw')}{2} + \frac{label({a_{k+1}})}{2} \hspace*{1cm} \mbox{ (by definition of CGR)} \\
    &= \frac{p_{last}(u)}{2^{k+1}} + \frac{p_{last}(w')}{2}+ \frac{label({a_{k+1}})}{2} \hspace*{2cm} \mbox{ (by I.H.)}  \\
    &= \frac{p_{last}(u)}{2^{k+1}} + p_{last}(w)  \hspace*{2cm}  \mbox{ (by definition of CGR)}
\end{align*}

\end{conditionalproof}

%% file: Proofs/Section3/prop_boundary.tex
\begin{proposition}
\label{prop:boundary}
    Let $s$ be a sequence of length $n$ over $\Sigma$, and  let $1\leq k \leq n$. The last point $p_{last}(s)$  of $CGR(s)$ is not on the boundary of a cell $c(w)$  associated with any $k$-mer $w$ in  $\Sigma^k$.  
\end{proposition}

\begin{conditionalproof}
   First note that any point in $CGR(s)$ of a sequence $s$  is strictly inside the CGR square.  Indeed, for any $s' \in \Sigma^n$, the $x$-coordinate of  $p_{last}(s') = (x_{s'}, y_{s'})$, the last point of $CGR(s')$, is by definition of the form
   \[x_{s'} = \frac{\pm 1}{2} + \frac{\pm 1}{2^2} + \frac{\pm 1}{2^3} + \cdots + \frac{\pm 1}{2^n}.\]
  
    It follows that the absolute values of $x_{s'}$ and $y_{s'}$ are both strictly less than $1$, therefore, $p_{last}(s')$ is strictly inside the CGR square. 
    
    Let $s=s'w$ where $w$ is of length $k$. Then,  by Lemma~\ref{lem:decom},
    \[p_{last}(s) = \frac{p_{last}(s')}{2^k} + p_{last}(w).\]
    The proof is completed by the fact that the absolute values of $x_{s'} $ and $y_{s'}$  are strictly less than 1.
    
\end{conditionalproof}

%% file: Proofs/Section3/lem_inclusion.tex
\begin{lemma}~\label{lem:inclusion}
    Let $w_1$, $w_2$ be words in $\Sigma^+$.   Then $c(w_1 w_2)  \subseteq c(w_2)$.
\end{lemma}
\begin{conditionalproof}
Let $(x,y) \in  c(w_1w_2)$,  let $p_{last}(w_1w_2) = (x_{w_1 w_2},y_{w_1 w_2})$, $p_{last}(w_1) = (x_{w_1},y_{w_1})$ and let $p_{last}(w_2) = (x_{w_2},y_{w_2})$. Without loss of generality, we consider the $x$-coordinates of all points. By the definition of $c(w_1w_2)$, we have \[-1/2^{|w_1 w_2|} < x-x_{w_1 w_2} < 1/2^{|w_1 w_2|}.\]

By Lemma~\ref{lem:decom},  we have that
    \[ x_{w_1 w_2} = \frac{x_{w_1}}{2^{|w_2|}} +  x_{w_2}.\]
    By replacing $x_{w_1 w_2}$ into the previous  inequality  we obtain
    \[-1/2^{|w_1w_2|} + x_{w_1}/2^{|w_2|}< x-x_{w_2} <   1/2^{|w_1w_2|} + x_{w_1}/2^{|w_2|} .\]
   Since $x_{w_1} \leq \sum_{l=1}^{|w_1|}2^{l-1}/2^{|w_1|} = (2^{|w_1|}-1) / 2^{|w_1|}$, we have $x - x_{w_2} < 1/2^{|w_1w_2|} + x_{w_1}/2^{|w_2|} \leq 1/2^{|w_2|}$. Similarly, we have $-1/2^{|w_2|} \leq -1/2^{|w_1w_2|} + x_{w_1}/2^{|w_2|} < x - x_{w_2}  $. Therefore, $-1/2^{|w_2|} < x -x_{w_2} < 1/2^{|w_2|}.$
   
   Since an analogous argument can be made for the $y$-coordinate, it follows that the arbitrary point  $(x, y)$ from $c(w_1 w_2)$ is  situated in $c(w_2)$.

\end{conditionalproof}

%% file: Proofs/Section3/prop_cell_hierarchy.tex
\begin{proposition}
\label{prop:cell_hierarchy}
  Let $k \geq 1$ and $n\geq 1$. For a given $k$-mer $w \in \Sigma ^k$ we have that $\overline{c(w)} = \bigcup_{w' \in \Sigma^n}\overline{c(w' w)}$.   
    
\end{proposition}
\begin{conditionalproof}
    By strong induction on $n$. 
    
    \noindent
    {\bf Base Case.} For $n=1$, we have to prove that 
    \[\overline{c(w)} = \bigcup_{a \in \Sigma}\overline{c(a w)}. \]
    Let $p_{last}(w) = (x_{w},y_{w})$ and let $label(a)=(x_a,y_a)$. By definition of $\overline{c(w)}$,
    \[\overline{c(w)} = \{(x,y)|\;\; -1/2^{|w|}\leq x-x_{w} \leq 1/2^{|w|},\; -1/2^{|w|}\leq y-y_{w} \leq 1/2^{|w|}\}.\]
    By Lemma~\ref{lem:decom} and the definition of  $\overline{c(a w)}$, we have
    \[\overline{c(a w)} = \{(x,y)|\;\; \frac{-1 + x_a }{2^{|w|+1}}\leq x-x_{w} \leq \frac{1 + x_a }{2^{|w|+1}},\;\; \frac{-1 + y_a }{2^{|w|+1}}\leq y-y_{w} \leq \frac{1 + y_a }{2^{|w|+1}}\}.\]

   \noindent
    Since $x_a, y_a \in \{ -1, 1 \}$, it suffices to consider all possible pairs in $ \{ -1, 1 \} \times \{ -1, 1 \}$. 
    Substituting each of these pairs into the preceding equation and taking their union, leads to the desired range for points in $\overline{c(w)}$, from which it follows that $\overline{c(w)} = \bigcup_{a \in \Sigma} \overline{c(a w)}.$


    \noindent
    {\bf Inductive Step.}  Assume  now that the statement holds for all $1\leq i \leq j$, and  show that it holds for $j+1$.
    Note that
    \[ \bigcup_{w' \in \Sigma^{j+1}}\overline{c(w' w)} = \bigcup_{a \in \Sigma} \bigcup_{w'' \in \Sigma^j}\overline{c(w'' a w)}.\]
    Since  by the induction hypothesis, we have $\bigcup_{w'' \in \Sigma^j}\overline{c(w'' a w)} = \overline{c(a w)}$,  and also that $\overline{c(w)} = \bigcup_{a \in \Sigma}\overline{c(a w)}$, the proof of the inductive step  is complete.
\end{conditionalproof}

%% file: Proofs/Section3/thm_main_thm_fcgr.tex
\begin{theorem}
\label{thm:main_thm_fcgr}
    Given a DNA sequence $s \in \Sigma^n$, and a $k$-mer $w \in \Sigma^k$, where  $1 \leq k \leq n$,   we have that  $$ occ(s, w) = \mbox{card}(CGR(s) \cap c(w)) $$ where   $occ(s, w)$  denotes the number of occurrences of the $k$-mer $w$  in $s$. 
\end{theorem}
\begin{conditionalproof}
   Let $k\geq 1$ and let $w \in \Sigma^k$ be a $k$-mer.  We first prove that 
     $occ(s, w) \leq \mbox{card}(CGR(s) \cap c(w))$. 
     
     Consider an arbitrary occurrence of $w$ in $s$. Then  $s = s_1 w s_2$ for some $s_1,s_2\in \Sigma^*$, and let us denote  $s' = s_1w$. 
     
     If $s_1 = \lambda$, then  $p_{last}(s') = p_{last}(w)$ is the center of $c(w)$ and thus is located  inside $c(w)$.

     If $s_1 \in \Sigma^+$, by applying Lemma~\ref{lem:decom} to $s'$, we have that the last point $p_{last}(s')$ of $CGR(s')$ is
    \[p_{last}(s') =  p_{last}(s_1 w) = \frac{p_{last}(s_1)}{2^k} + p_{last}(w).\] From the proof of Proposition~\ref{prop:boundary}, we have that  $p_{last}(s_1)$ is located  inside the CGR square, that is,  $p_{last}(s_1) \in \{(x,y)|\; -1<x<1,-1<y<1\}$. 
This implies  that $$p_{last}(s')  = p_{last}(s_1 w) \in \{p_{last}(w) + (x,y)|\;  -1/2^k<x<1/2^k,-1/2^k<y<1/2^k\}.$$
By comparing this with  Definition~\ref{def:cell_of_k_mer}  of  $c(w)$,   Definition~\ref{def:cell_of_k_mer}, it follows that $p_{last}(s')$ is located  inside $c(w)$.     
    
    Since we reached the same conclusion in both possible cases, it follows that any occurrence of $w$  in $s = s_1 w s_2$ corresponds to a CGR  point $p_{last}(s_1 w) $  located inside $c(w)$. From this, 
    It follows that $occ(s, w) \leq \mbox{card}(CGR(s) \cap c(w)) $.

    We now prove  that  $occ(s, w) \geq \mbox{card}(CGR(s) \cap c(w))$. 
    
   Consider a point of $CGR(s)$ that is located inside $c(w)$. This point is $p_{last}(s')$, the last point of $CGR(s')$, for some prefix $s'$ of $s$. 
   
   We first show that   $|s'| \geq k$. 
   Assume, for the sake of contradiction, that   $|s'| = k'< k$. 
   
   If $s'$ is not a proper suffix of $w$, then for any decomposition  $w = w_1w_2$ with $|w_2| = |s'|$, we have that $w_2 \neq s'$. Then $p_{last}(s')$ will be located inside $c(s')$, which does not overlap with $c(w_2)$, hence $p_{last}(s')$  is not located inside $c(w_2)$. By Lemma~\ref{lem:inclusion}, $c(w) = c(w_1 w_2) \subseteq  c(w_2)$. This implies that  $p_{last}(s')$  is not located inside $c(w)$ - a contradiction.
 Thus, it must be the case that if $s' = k' < k$, then  $s'$ is a proper suffix of $w$, that is, $s' = w_2$ with $w = w_1 w_2$, and $w_1 \neq \lambda$.
   
However, this leads to a contradiction, as follows. From Lemma~\ref{lem:decom}, we have $p_{last}(w_2) - p_{last}(w_1w_2) = \frac{p_{last}(w_1)}{2^{k'}}$. For $p_{last}(s)$ to be located inside $c(w)$, we would need $2^{|w_1|} \cdot  p_{last}(w_1) \in(-1,1) \times (-1,1)$, which cannot happen for any $w_1 \neq \lambda$. Indeed,  for the aforementioned relation to hold, the  $x$-coordinate of  $p_{last}(w_1) = (x_{w_1}, y_{w_1})$  should satisfy 
\[2^{|w_1|} \cdot x_{w_1}  = \sum_{l=1}^{|w_1|} x_l \cdot 2^{l-1}.\]
If  $x_{w_1} = 1$,  then we have that  $2^{|w_1|} \cdot x_{w_1} \geq 1$, and if   $x_{w_1} = -1$, then  we have that $2^{|w_1|} \cdot x_{w_1} \leq -1$. In both situations, the sum on the right-hand side cannot be in the interval $(-1, 1)$,  while the quantity on the left-hand side is in  $(-1, 1)$  -   a contradiction.
    
    Thus, it follows that $k \leq |s'| \leq n$, where  $s'$ is a prefix of $s$. Let $s'=s_1w'$ where $s_1 \in \Sigma^*, w' \in \Sigma^k$.  If $w' \neq w$, using the same argument as before, one can show that $p_{last}(s')$ will be in $c(w')$, which contradicts the fact that cells associated with different $k$-mers are non-overlapping.   Therefore $s' = s_1 w$, which means that the point in $CGR(s)$ that we considered, $p_{last}(s')$, corresponds to the specific occurrence of $w$ in $s$.   Since distinct $s'$ result in distinct $p_{last}(s')$, each point in  $CGR(s)$ located inside $c(w)$  corresponds to a distinct occurrence of $w$ in $s$. This further  implies that
    $ \mbox{card}(CGR(s) \cap c(w)) \leq occ(s, w)$, which completes the proof.
\end{conditionalproof}

%% file: Proofs/Section3/corollary_cgr2fgr.tex
\begin{corollary}
\label{corollary:cgr2fgr}
 Let $s \in \Sigma^n$ be a sequence, and let $1\leq k \leq n$ be a $k$-mer length.  Computing  $FCGR_k(s)$ by discretizing $CGR(s)$  into grid cells of order $k$  as defined in Definition \ref{def:grid_def} (i.e., at resolution $2^k \times 2^k$)  is equivalent to counting the number of occurrences of  $k$-mers  in the sequence $s$.
\end{corollary}
\begin{conditionalproof}
It follows from Theorem~\ref{thm:main_thm_fcgr}, Proposition \ref{prop:correspondence} and the bijection that it establishes.
\end{conditionalproof}

%% file: Sections/k-mers.tex
\section{From $k$-mer Distributions to Synthetic DNA}\label{sec:k_mer_freq}

 In this section, we explore the concept of empirical $k$-mer distribution of a sequence $s$, and formalize the correspondence with its $FCGR_k(s)$. We exploit this correspondence to generate synthetic DNA sequences and their respective CGRs by tracing an Eulerian path on the De Bruijn multigraph built from a target $k$-mer distribution. The main result (Theorem \ref{thm:reconstruction}) shows that, for any target distribution on the standard probability simplex satisfying the structural constraints imposed by the linear nature of DNA sequences, we can construct a synthetic sequence $s$ whose empirical $k$-mer distribution matches the target with arbitrarily high precision.  A computational tool implementing the algorithm is available at~ \href{https://tinyurl.com/kmer2cgr}{\url{https://tinyurl.com/kmer2cgr}}, and is used to empirically evaluate the algorithm's practical reconstruction accuracy through computational experiments using both real and synthetic target distributions.

\begin{definition}
Given the DNA alphabet \(\Sigma = \{A,C,G,T\}\) and a $k$-mer \( w=a_1\dots a_k \in \Sigma^k \), we define the index function $\text{idx}:\Sigma^k \rightarrow \{0,4^k-1\}$ as:
\begin{equation}\label{eq:kmer2idx}
\text{idx}(w) = \sum_{t=1}^{|w|} 4^{|w|-t}\times \xi(a_t),
\end{equation}
where the labelling function $\xi: \Sigma \rightarrow \{0,1,2,3\}$ assigns each nucleotide to an integer as follows:
$$
\xi(A) = 0,\; \xi(C) = 1,\; \xi(G) = 2,\; \xi(T) = 3.
$$
\end{definition}





 The previous definition interprets each $k$-mer $w$ as a base-4 numeral via the digit assignment $\xi(\cdot)$. Note that $\text{idx}(w)$ is a bijection between $\Sigma^k$ and $\{0,4^k-1\}$, with $k=|w|$. This bijection induces a natural ordering on $\Sigma^k$, where for $u, v \in \Sigma^k$, we write $u < v$ if and only if $\mathrm{idx}(u) < \mathrm{idx}(v)$. Based on the label assignment  $\xi(\cdot)$, this ordering corresponds to the lexicographic order given by the nucleotide ordering $A < C < G < T$.

\begin{definition}[$k$-mer frequency vector] Given a DNA sequence $s$, its $k$-mer frequency vector $F_k(s)$ is a vector in $\mathbb{N}^{4^k}$ defined as: 
$$F_k(s)= (occ(s, w_0), \ occ(s,w_1), \dots, \ occ(s, w_{4^k-1}))$$ 
where $w_i = \text{idx}^{-1}(i)$ is the  \(k\)-mer corresponding to index $i$ under the mapping in \eqref{eq:kmer2idx}.
\end{definition}

We now formalize the intuition that the $k$-mer frequency vector $F_k(s)$ corresponds to a vectorization of the $FCGR_k(s)$ matrix and compute the mapping between the matrix coordinates and the vector positions. For that purpose, we use the $label(\cdot)$ function (Definition~\ref{def:cgr_sqr}) that links the coordinates of each corner in the CGR square with each nucleotide. Recall that for each $k$-mer $w = a_1\cdots a_k$,  $label(a_l) = (x_l, y_l)$, with $x_l,y_l \in \{-1,1\}$ for all $ 1\leq l \leq k$. Consequently, $label^{-1}(x_l, y_l) = a_l$.


\input{Proofs/Section4/lem_fcgr_idx2kmer}

\input{Proofs/Section4/prop_fcgr_idf2lexi}

{Proposition~\ref{prop:fgr_idf2lexi} confirms the intuition that the $k$-mer frequency vector $F_k(s)$ corresponds to a vectorization of the $FCGR_k(s)$ matrix. Furthermore,  we now have a connection 
between each word $w$, its corresponding point in the CGR ($p_{last}(w)$), the indices $(i, j)$ in the FCGR matrix {(Definition~\ref{def:grid_def})  and its position  in the $k$-mer frequency  vector given by the lexicographic order, $\text{idx}(w)$. 

We will now proceed to define a probabilistic framework for comparing the $k$-mer composition of sequences of different lengths,  with the ultimate purpose of generating synthetic DNA sequences from given $k$-mer frequency vectors. 

\begin{definition}[$k$-mer distribution]
For a sequence $s$ of length $n$, the empirical $k$-mer distribution  of $s$ is defined as the probability vector 
$$\vtheta_s = (\theta_0^s, \theta_1^s, \dots, \theta_{4^k-1}^s) =\frac{1}{n-k+1} \cdot F_k(s).$$
\end{definition}

In other words, $\theta_i^s = \frac{occ(s, w_i)}{n - k + 1}$ represents the normalized frequency of the $k$-mer  $w_i$ in $s$, where $w_i$ denotes the $i$-th $k$-mer in the lexicographic order.  Given $\vtheta \in \Delta^{4^k-1}$ and a $k$-mer $w$, we denote the component $\theta_{\text{idx}(w)}$ by $\theta_w$. 

Note that any $\vtheta_s$  is a point in the \(4^k\)-dimensional probability simplex \(\Delta^{4^k-1}\) and represents the empirical distribution of $k$-mers in $s$. However, not all points $\vtheta \in \Delta^{4^k-1}$ correspond to an empirical $k$-mer distribution; in particular, the points with at least one irrational coordinate do not correspond to an empirical $k$-mer distribution. 

Besides the normalization constraint, if $\vtheta_s \in \Delta^{4^k-1}$ is a valid $k$-mer distribution, every internal $(k-1)$-mer is shared by two overlapping $k$-mers. For example, every occurrence of a $k$-mer $w=a_1a_2\dots a_k$ must be followed sequentially by a $k$-mer $w'=a_2a_3\dots a_{k+1}$.  This can be formalized by the requirement  $\textsc{suffix}(w)=\textsc{prefix}(w')$, where $\textsc{prefix} (w)$ (respectively $\textsc{suffix}(w)$) is defined as being the word consisting of the first (respectively last)  $|w|-1$ letters in $w$. This condition imposes a marginalization constraint on all the $k$-mers, e.g. $$\theta_{AXY} + \theta_{CXY} + \theta_{GXY} + \theta_{TXY} = \theta_{XYA} + \theta_{XYC} + \theta_{XYG} + \theta_{XYT}  =\theta_{XY}.$$
Thus, for each $(k-1)$-mer $v \in \Sigma^{k-1}$ (except edge cases in finite sequences), the sum of the frequencies of $k$-mers with $v$ as a prefix must equal the sum of the frequencies of $k$-mers with $v$ as a suffix. When considered across all $k$-mers, this condition imposes a set of additional marginal constraints that imply that $\vtheta_s$ does not occupy the full $4^k$-dimensional simplex but rather a lower-dimensional sub-polytope \cite{Nissen_taxonomic_clustering_VAMB, Kislyuk2009}. Formally, for every \( v \in \Sigma^{k-1} \),
\begin{equation}\label{eq:marginal_constraint}
\sum_{a \in \Sigma} \theta_{v a}^{s} = \sum_{a \in \Sigma} \theta_{av}^{s} = \theta_v^s .
\end{equation}
Note that the constraints in~\eqref{eq:marginal_constraint} represent necessary conditions for a point  $\vtheta \in \Delta^{4^k-1}$ to be the $k$-mer distribution of a valid DNA sequence. We will later show that these constraints are a sufficient condition for sequence reconstruction. 

Given a $k$-mer distribution $\vtheta$, our objective is to reconstruct a sequence $s$, whose empirical $k$-mer frequency vector approximates the target distribution $\vtheta$ as closely as possible. Then, we can compute the CGR according to Definition~\ref{def:CGR} to visualize the patterns corresponding to specific $k$-mer over- or under-representation. A similar problem has been studied in the field of genome assembly \cite{MILLER2010315}, where De~Bruijn‑graph–based methods have proven effective for assembling genomes from fragmented reads \cite{Pevzner_Eulerian_2001, SPADES_Pevzner}. We will show that a suitable traversal of an appropriately constructed De~Bruijn graph \iftoggle{arxiv}{(See Appendix~\ref{sec:debruijn})}{}, combined with corrections for rounding error and violation of \eqref{eq:marginal_constraint}, produces a sequence that satisfies an approximation criterion.  

\begin{definition}
A directed multigraph $M$ is a pair $(V, E)$, where $V$ is a set of nodes, and $E$ is a multiset of ordered pairs from $V \times V$, called edges. Loops and parallel edges are permitted. 
\end{definition}

\begin{definition}
\label{def:db_multigraph}
Let $k \geq 2$, $n > k$, and $s \in \Sigma^n$. The De Bruijn multigraph, denoted $\textsc{De Bruijn}^{+}_k(s)$, is  the directed multigraph $\bigl(V_{k-1}(s),\,E_k(s)\bigr)$, 
where
\[
V_{k-1}(s)
\;=\;
\bigl\{\,v\in\Sigma^{k-1}\;\big|\;v\text{ appears as a substring of }s \bigr\},
\]
\begin{equation*}
\label{def:DeBruijn+}
\begin{split}
E_k(s) = \{ (u,v) \in V_{k-1}(s)^2 \mid \text{there exists a } k\text{-mer } w \in \Sigma^k \text{ that appears in } s \\
\text{ with } \textsc{Prefix}(w)=u \text{ and } \textsc{Suffix}(w)=v \}.
\end{split}
\end{equation*}
In particular, for each \((u,v)\in V_{k-1}(s)^2\), the multiplicity of the edge \((u,v)\) in \(E_k(s)\) is exactly $occ(s,w)$. 
\end{definition}

Based on Definition~\ref{def:DeBruijn+}, $\textsc{De Bruijn}^{+}_k(s)$ is completely determined by the collection of $k$-mers of $s$, so it can be constructed directly from a $k$-mer  frequency vector without explicit knowledge of the sequence $s$. With this framework in place, assume a target distribution \(\vtheta_s\). It is then possible to rescale \(\vtheta_s\) by \(n-k+1\) to obtain an array of pseudo counts $\vc_s$ such that each component is defined by \( c_w^{s} = (n-k+1) \cdot \theta_w^{s} \). Consequently, the problem of string reconstruction reduces to finding a path in the De Bruijn multigraph (constructed from these pseudo-counts) that traverses every edge exactly once, i.e.,  an Eulerian path. Euler's theorem provides the necessary and sufficient conditions for the existence of such a path.

\begin{theorem}{(\textit{Euler})}
Let $G = (V, E)$ be a directed graph, and for any vertex $v \in V$ let $\textsc{in}(v)$ and $\textsc{out}(v)$ denote its indegree and outdegree, respectively. Then $G$ contains an \textit{Eulerian cycle} if and only if:
(1)  $G$ is strongly connected, i.e., there exists a directed path between any pair of vertices in $G$, and (2) for every vertex \( v \in V \), it holds that $\textsc{IN}(v) = \textsc{OUT}(v)$.

\end{theorem}
Thus, for a directed graph $G$ to contain an Eulerian cycle, it must be strongly connected and balanced at every vertex. In our setting, this implies that the hypothetical counts $c_w$ must satisfy the marginal constraint on the de Bruijn graph to yield a valid DNA sequence $s$. In particular, if we associate each $(k-1)$-mer $v$ with the $k$-mers that have \( v \) as a prefix or suffix, then the following condition in~\eqref{eq:marginal_constraint} must hold for every vertex \( v \), except for the unique vertices corresponding to the start and end of an Eulerian path. We now formalize these conditions in the following theorem:

\input{Proofs/Section4/thm_reconstruction}

We validate the theoretical construction through numerical experiments, using two different approaches. First, a fragment of 100,000 base pairs (bp) is sampled from the reference genome of 100 different species, each representing a distinct kingdom of life. For each sampled fragment, empirical $k$-mer distributions are computed for $k\in\{2,\dots,6\}$, and the average reconstruction error, measured by the total variation of $ k$-mer distances, is determined for each kingdom. These results are summarized in Figure~\ref{fig:Template-based_reconstruction}. Second, we sample arbitrary distributions from the relevant sub-polytope of $\Delta^{4^k-1}$. To this end, we employ a variation of the hit–and–run sampling procedure \cite{hit_and_run}\iftoggle{arxiv}{(See Appendix \ref{appendix:hit-and-run})}{}, which operates within the null-space defined by the linear constraints in~\eqref{eq:marginal_constraint}. Here, the sequence length is fixed at $n_{\min}=2\cdot4^k/{\epsilon}+k-1$  for $\epsilon=0.01$, and sampling is performed for $k \in \{2, \dots, 6\}$. Figure~\ref{fig:de-novo_reconstruction} illustrates the reconstruction error for different values of $k$ over points sampled from the desired sub-polytope. 

\begin{figure}[!t]
    \centering
    \begin{subfigure}[b]{0.48\linewidth}
        \centering
        \includegraphics[width=\linewidth]{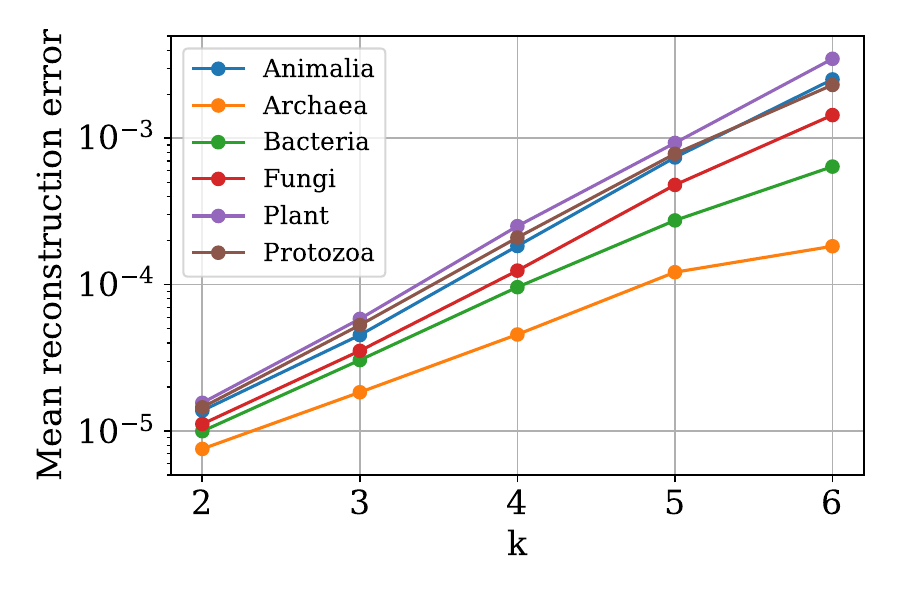}
        \caption{}
        \label{fig:Template-based_reconstruction}
    \end{subfigure}
    \hfill
    \begin{subfigure}[b]{0.48\linewidth}
        \centering
        \includegraphics[width=\linewidth]{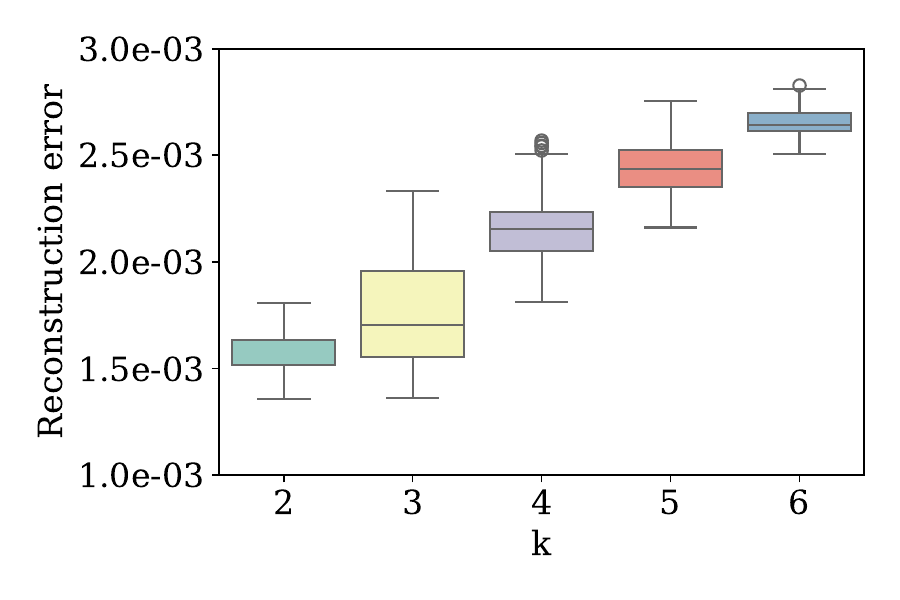}
        \caption{}
        \label{fig:de-novo_reconstruction}
    \end{subfigure}
    \caption{{\bf (a) Reconstruction of sequences from $k$-mer distributions of real DNA sequences.} The reconstruction is computed across different values of $k$, and sequences spanning species from each of the six kingdoms of life. For each kingdom, 100 genomes were selected, and a single 100,000-long DNA fragment was randomly selected from each genome. The mean error was computed for each pair ($k$, kingdom). {\bf (b) Reconstruction of sequence from arbitrary $k$-mer distributions that satisfy the marginal constraints}. Each point is sampled from the standard probability simplex $\Delta^{4^k-1}$ subject to the constraints in~\eqref{eq:marginal_constraint} for different values of $k \in \{2, \dots, 6\}$. The sequence length $n_{\min}$, calculated for $\epsilon=0.01$, ranges from 3,200 for $k=2$, to 819,200 for $k=6$. The reconstruction error is less than $\epsilon$ for all the values of $k$.}
    \label{fig:debruijn_example}
\vspace{-5mm}
\end{figure}

The results of our experiments indicate that the reconstruction algorithm is capable of recovering sequences whose $k$-mer distributions closely approximate the target distributions, regardless of whether the target is sampled from a reference DNA sequence or sampled from the simplex. A computational tool  implementing the algorithm  is available at~ \href{https://tinyurl.com/kmer2cgr}{\url{https://tinyurl.com/kmer2cgr}} (see \autoref{fig:web-app}).

\begin{figure}[!h]
    \centering
    \begin{subfigure}[b]{0.48\linewidth}
        \centering
        \includegraphics[width=0.5\linewidth]{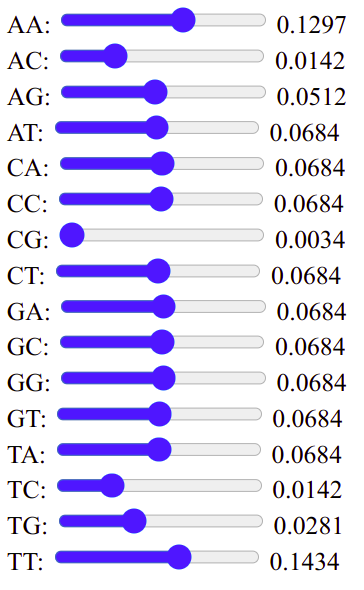}
        \caption{}
        \label{fig:settings_web_app}
    \end{subfigure}
    \begin{subfigure}[b]{0.48\linewidth}
        \centering
        \includegraphics[width=0.8\linewidth]{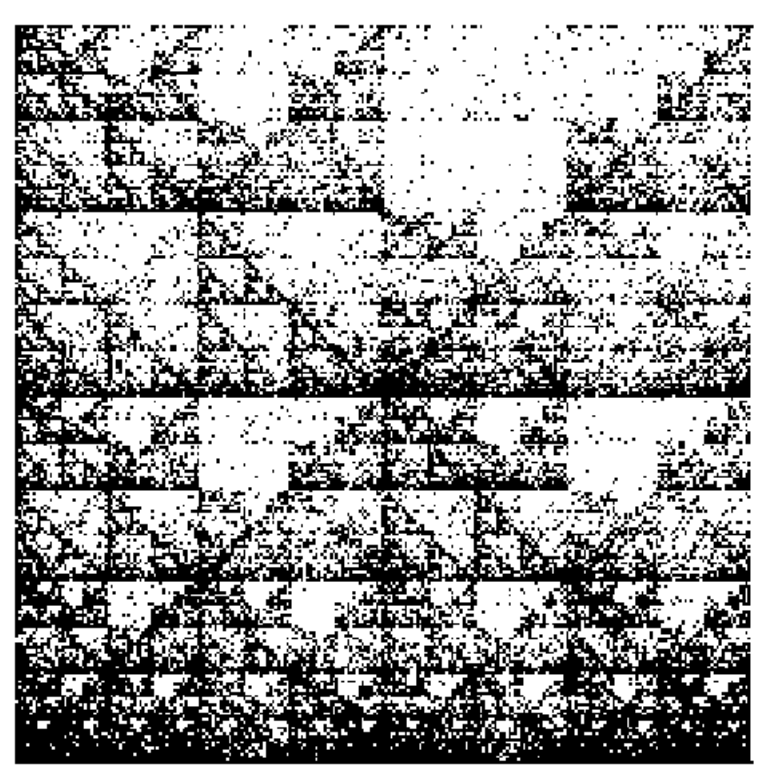}
        \caption{}
        \label{fig:web_app_result}
    \end{subfigure}
    \caption{Snapshot of the functionality of the computational tool: \textbf{(a)} Input: Dinucleotide  ($k=2$) distribution selected by a user via interactively adjusting 16 different sliders, each corresponding to one dinucleotide. \textbf{(b)} Output: CGR of a reconstructed  DNA sequence whose $k$-mer distribution closely approximates the $k$-mer distribution in (a).}
    \label{fig:web-app}
\vspace{-5mm}
\end{figure}





%% file: Proofs/Section4/lem_fcgr_idx2kmer.tex
\begin{lemma}\label{lem:fcgr_idx2kmer}
Let $k \ge 1$,  and  $0 \leq i, j \leq  2^k -1$. If $i = \sum_{l=0}^{k-1}\alpha_{l}2^{l}$ and $j = \sum_{l=0}^{k-1}\beta_{l}2^{l}$,  with $\alpha_{l}, \beta_l \in \{0,1\}$ are the binary expansions of $i$ and $j$, then $cell_k(i,j) = c(a_1\cdots a_k)$ where  $c(a_1 ... a_k)$ is the cell associated with $a_1  ... a_k$ constructed as 
\begin{equation}\label{eq:fcgr_idx_2_matrix}
a_l = label^{-1}(1-2\alpha_{l-1},2\beta_{l-1}-1), \mbox{ for } 1\leq l \leq k
\end{equation}

\end{lemma}

\begin{conditionalproof}
From Proposition~\ref{prop:correspondence}, we know that $$j = \frac{2^k+1 + \sum_{l=1}^k x_l \cdot 2^{l-1} }{2}, \quad i = \frac{2^k +1 - \sum_{l=1}^k y_l\cdot 2^{l-1} }{2},$$ which can be rewritten as:

$$j = \frac{1}{2} \left(\sum_{l=1}^k 2^{l-1} \cdot (1 + x_l) + 1\right), \quad i = \frac{1}{2} \left(\sum_{l=1}^k 2^{l-1} \cdot (1 - y_l) + 1\right).$$
Now, considering the binary expansion of the indices, we get the following equations:
$$j = \sum_{l=0}^{k-1}\beta_l \cdot 2^t  = \sum_{l=1}^k 2^{l-1} \cdot \frac{(1 + x_l)}{2}, \quad i = \sum_{l=0}^{k-1}\alpha_l \cdot 2^l = \sum_{l=1}^k 2^{l-1} \cdot \frac{(1 - y_l)}{2}, $$
which we solve for $x_l$ and $y_l$ to obtain the correspondence with the coefficients in the binary expansion:
\begin{equation}
  x_l = 1 - 2\alpha_{l-1}, \quad y_l=2\beta_{l-1}-1  
\end{equation}

Finally, by Definition~\ref{def:cgr_sqr} we get that: $a_l = label^{-1}(x_l,y_l) = label^{-1}(1 - 2\alpha_{l-1}, 2\beta_{l-1}-1)$
\end{conditionalproof}

%% file: Proofs/Section4/prop_fcgr_idf2lexi.tex
\begin{proposition}
\label{prop:fgr_idf2lexi}
\sloppy
Let $k \ge 1$, $n \ge 1$, $s \in \Sigma^n$ and $0 \leq i, j \leq  2^k -1$. If $i =  \sum_{l=0}^{k-1}\alpha_{l}2^{l}$ and $j = \sum_{l=0}^{k-1}\beta_{l}2^{l}$, with $\alpha_{l}, \beta_l \in \{0,1\}$ (their binary expansions), then $FCGR_k(s)(i,j)= occ(s, w_\tau)$, where
\begin{equation}
\tau = \sum_{l=0}^{k-1}4^{k-l-1}\times\xi(label^{-1}(1-2\alpha_{l-1},2\beta_{l-1}-1))
\end{equation}
\end{proposition}

\begin{conditionalproof}
First, note the correspondence between each $k$-mer and its position in the $ FCGR_k$ matrix. In Lemma~\ref{lem:fcgr_idx2kmer}, we have found the correspondence between the indices in the matrix (\eqref{eq:fcgr_idx_2_matrix}) so that each position \((i,j)\) defines a \(k\)-mer $w=a_1a_2\cdots a_k$. Replacing~\eqref{eq:fcgr_idx_2_matrix} into~\eqref{eq:kmer2idx}, yields the corresponding index $\tau$ of $w$ in the lexicographic order.
\end{conditionalproof}

%% file: Proofs/Section4/thm_reconstruction.tex
\begin{theorem}\label{thm:reconstruction}
Let  $1 < k < n\in \mathbb{Z}$, $\epsilon \in [0,1]$, and let $\vtheta$ be a point in \( \Delta^{4^k-1} \) satisfying the marginal constraint in~\eqref{eq:marginal_constraint}
for every vertex $v$. If $n > \frac{2(k-1)\cdot 4^{k}}{\epsilon} + k - 1,$ then there exists an algorithm that constructs a sequence $s$ whose empirical $k$-mer distribution $\hat{\vtheta}_s = \frac{F_k(s)}{\|F_k(s)\|_1}$ satisfies $\|\hat{\vtheta}_s - \vtheta\|_1 \leq \epsilon.$
\end{theorem}

\begin{conditionalproof}
Let $\vc$ be the rounded count vector obtained by rescaling $\vtheta$ by $n-k+1$ and let $G = (V,E)$ be the de Bruijn graph constructed from \(\vc\). For each vertex \( v \in V \) (corresponding to a \((k-1)\)-mer), define the flow imbalance
\[
\delta(v) = \sum_{a \in \Sigma} c_{va} - \sum_{a \in \Sigma} c_{av}.
\]
Denote by \( V^+ = \{v \in V : \delta(v) > 0\} \) the set of vertices with excess outgoing edges, and by \( V^- = \{v \in V : \delta(v) < 0\} \) the set of vertices with excess incoming edges. By conservation of the total edge count, \(\sum_{v \in V} \delta(v) = 0\); hence, the total excess outdegree equals the total excess indegree.

To balance the graph and satisfy Euler’s condition, we add artificial edges. Specifically, for every $v\in V^+$ and every $w\in V^-$ there is a directed path $P_{v\to w}$
of length at most $k-1$ edges (and hence $k-2$ intermediate vertices $u_1,\dots,u_{k-2}$). 
This can be seen from the De Bruijn multigraph associated with the sequence $vw \in \Sigma^{2k-2}$, it starts with vertex $v$ and ends with vertex $w$. To correct the imbalance of $G$ at $v$ and $w$, we add one edge connecting each node in $P_{v\to w}$, which increments the count of each edge in $P_{v\to w}$ by one.  Doing this $\delta(v)$ times for each $v\in V^+$ restores in‐ and out‐degree balance at every vertex (since intermediate nodes gain one in‐edge and one out‐edge, leaving their imbalance unchanged). Denote by $ G' = (V, E')$ the resulting balanced graph. By construction, every vertex in $G'$ now has equal indegree and outdegree. If $G'$ is strongly connected, Euler’s theorem guarantees the existence of an Eulerian cycle. If not, at most \(4^{k-1}\) additional artificial edges are needed to connect the components. The sequence $s'$ is reconstructed by traversing the Eulerian path in G'.

\medskip
\noindent
\textbf{Error bound:} Define \(\delta_{max} = \max_{v \in V} |\delta(v)|\). Since the maximum rounding error for any count is \(\frac{1}{2}\) and each \((k-1)\)-mer has at most 4 incoming and 4 outgoing $k$-mers, we have \(\delta_{max} \leq 4\). Let \( n_{art} \) denote the total number of artificial edges added; then
\[
n_{art} = \sum_{v \in V^+} (k-1) \cdot \delta(v) \leq (k-1) \cdot \delta_{max} \cdot |V^+| \leq (k-1)\cdot 4^{k}.
\]
Denote \(\vc = (c_1,\ldots,c_{4^k})\) the original $k$-mer counts and by \(\vc' = (c'_1,\ldots,c'_{4^k})\) the $k$-mer counts in the reconstructed sequence \( s' \). Each artificial edge increases the corresponding $k$-mer count by 1, consequently, the error in the $i$-th component of the empirical distribution is
\[
\left|\hat{\theta}_{i}^{s'} - \theta_i\right| = \left|\frac{c'_i}{|E| + n_{art}} - \frac{c_i}{|E|}\right|,
\]

which can be rewritten as:
   $$|\hat{\theta}_{i}^{s'} - \theta_i| = \left|\frac{c'_i \cdot |E| - c_i \cdot (|E| + n_{art})}{|E| \cdot (|E| + n_{art})}\right|$$
   $$= \left|\frac{c'_i \cdot |E| - c_i \cdot |E| - c_i \cdot n_{art}}{|E| \cdot (|E| + n_{art})}\right|$$
   $$= \left|\frac{|E| \cdot (c'_i - c_i) - c_i \cdot n_{art}}{|E| \cdot (|E| + n_{art})}\right|.$$


 When we sum this over all $k$-mers, using the facts that $\sum_{i=1}^{4^k} |c_i'-c_i| =n_{art}$ and $\sum_{i=1}^{4^k} c_i = |E|$,  we get:
   $$\sum_{i=1}^{4^k} |\hat{\theta}_{i}^{s'} - \theta_i| \leq \frac{2n_{art}}{|E| + n_{art}},$$

\[
\|\hat{\vtheta}_{s'} - \vtheta\|_1 \leq \frac{2n_{art}}{|E| + n_{art}}.
\]
Since \( |E| = n - k + 1 \) and $n_{art} \leq (k-1)\cdot 4^{k}$, it follows that
\[
\|\hat{\vtheta}_{s'} - \vtheta\|_1 \leq \frac{2\cdot(k-1)\cdot 4^{k}}{(n - k + 1) + (k-1) \cdot 4^{k}}.
\]

Imposing the condition $ n > \frac{2\cdot (k-1) \cdot4^{k}}{\epsilon} + k - 1,$ ensures that $\|\hat{\vtheta}_{s'} - \vtheta\|_1 \leq \epsilon.$
\end{conditionalproof}

%% file: Sections/Conclusion.tex
\section{Conclusions and Future Work}


This work formalizes the connections between CGR, $FCGR_k$, and $k$-mer frequency vectors as genomic signatures. We note, however, that each genomic signature has distinct features and that each has proven effective in specific applications, depending on various factors such as sequence length or sequence dataset size.  The aforementioned demonstrated connections between these types of genomic signatures could potentially aid in the determination of their suitability for particular practical applications. We also present an algorithm that generates synthetic DNA sequences and their corresponding CGRs from target $k$-mer frequency vectors. This methodology opens a new avenue for synthetic DNA data generation in contrastive machine learning pipelines for taxonomic classification applications, whereby data augmentations (training samples artificially generated  from existing data), are used to learn meaningful representations of the DNA sequences. Future work could explore the applicability of $k$-mer sampling for data augmentation and the effect of less stringent constraints, e.g., information-theoretic, on reconstruction guarantees.

%% file: Sections/Technical_Appendix.tex
\toggletrue{includeproofs}
\section{Proofs in Section 2}~\label{appendix:proofs_sec2}
\setcounter{theorem}{0}
\input{Proofs/Section2/theorem1}
\setcounter{corollary}{0}
\input{Proofs/Section2/cor_closure_of_CGR}
\setcounter{theorem}{1}

\input{Proofs/Section2/theorem2}
\section{Proofs in Section 3}~\label{appendix:proofs_sec3}
\setcounter{proposition}{0}
\input{Proofs/Section3/prop_correspondence}
\setcounter{lemma}{0}
\input{Proofs/Section3/lem_decom}
\setcounter{proposition}{1}
\input{Proofs/Section3/prop_boundary}
\setcounter{lemma}{1}
\input{Proofs/Section3/lem_inclusion}
\setcounter{proposition}{2}
\input{Proofs/Section3/prop_cell_hierarchy}
\setcounter{theorem}{2}
\input{Proofs/Section3/thm_main_thm_fcgr}
\setcounter{corollary}{1}

\input{Proofs/Section3/corollary_cgr2fgr}
\section{Proofs in Section 4}~\label{appendix:proofs_sec4}
\setcounter{proposition}{3}
\input{Proofs/Section4/lem_fcgr_idx2kmer}
\setcounter{lemma}{2}
\input{Proofs/Section4/prop_fcgr_idf2lexi}
\setcounter{theorem}{4}

\input{Proofs/Section4/thm_reconstruction}

%% file: Sections/General_Appendix.tex

\section{Dihedral Group}\label{subsec:dihedral_group}

The {\it dihedral group} \cite{Abstractalgebra} $D_{8}=\left\{e, r, r^{2}, r^{3}, s, s r, s r^{2}, s r^{3}\right\}$ is the group of all symmetries of the square, a group comprising rotations and reflections about lines joining midpoints of opposite sides, respectively reflections about diagonals: 
\begin{itemize}
  \item $\boldsymbol{e}$ is the identity (no change)
  \item $\boldsymbol{r}$  (respectively $r^{2}, r^{3}$)  is the rotation by $\frac{\pi}{2}$ (resp. $\pi$, $\frac{3 \pi}{2}$) counterclockwise.
  \item $\boldsymbol{s}$ is the reflection about  the line bisecting the left and right sides of the square (i.e. the line $y_{1}=\{(x, 0): x \in \mathbb{R}\}$ shown in Figure \ref{fig:cgr-permutations}.)
  \item $\boldsymbol{s r}$ is  the reflection about the diagonal from the bottom-right to the top-left vertex (i.e. the line $y_{2}=\left\{\left(x_{1}, x_{2}\right) \in \mathbb{R}^{2}: x_{1}=-x_{2}\right\}$ shown in Figure \ref{fig:cgr-permutations}.)
  \item $\boldsymbol{s r^{2}}$ is the reflection about the line bisecting the top and bottom sides of the square (i.e. the line $y_{3}=\{(0, y): y \in \mathbb{R}\}$ shown in Figure \ref{fig:cgr-permutations}.)
  \item $\boldsymbol{s r^{3}}$ is the reflection about the diagonal from the top-right to the bottom-left  vertex (i.e. the line $\left.y_{4}=\left\{\left(x_{1}, x_{2}\right) \in \mathbb{R}^{2}: x_{1}=x_{2}\right\}\right)$ shown in Figure \ref{fig:cgr-permutations}.)
\end{itemize}

\noindent
The following representations, using matrices in $M_{2 \times 2}(R)$, will be used to implement these symmetries.\\

$
\begin{aligned}
& \text {  } \boldsymbol{e} \mapsto\left[\begin{array}{ll}
1 & 0 \\
0 & 1
\end{array}\right] 
& \text {  } \boldsymbol{r} \mapsto\left[\begin{array}{cc}
0 & -1 \\
1 & 0
\end{array}\right] 
& \text { } \boldsymbol{r}^{2} \mapsto\left[\begin{array}{cc}
-1 & 0 \\
0 & -1
\end{array}\right]
& \text {  } \boldsymbol{r}^{3} \mapsto\left[\begin{array}{cc}
0 & 1 \\
-1 & 0
\end{array}\right] \\
& \text {  } \boldsymbol{s} \mapsto\left[\begin{array}{cc}
1 & 0 \\
0 & -1
\end{array}\right] 
& \text {  } \boldsymbol{sr} \mapsto\left[\begin{array}{cc}
0 & -1 \\
-1 & 0
\end{array}\right] 
& \text {  } \boldsymbol{sr^2} \mapsto\left[\begin{array}{cc}
-1 & 0 \\
0 & 1
\end{array}\right] 
& \text {  } \boldsymbol{sr^3} \mapsto\left[\begin{array}{cc}
0 & 1 \\
1 & 0
\end{array}\right]
\end{aligned}
$

\begin{figure}[htbp]
  \centering
  \setlength{\tabcolsep}{1pt}\renewcommand{\arraystretch}{0}
  \begin{tabular}{cccc}
    \subcaptionbox{\scriptsize$e\;()$}[.20\linewidth]{\includegraphics[width=\linewidth]{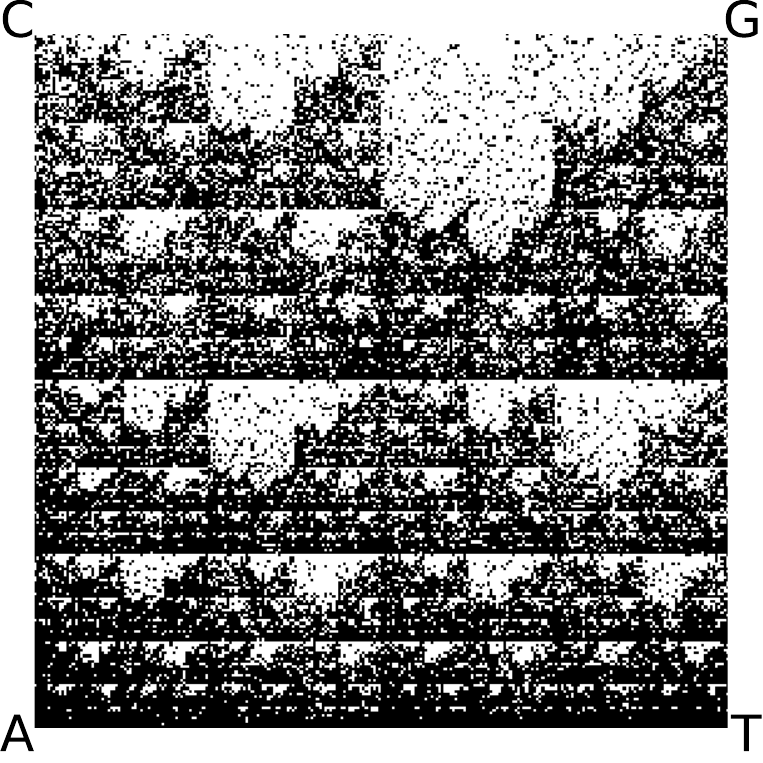}} &
    \subcaptionbox{\scriptsize$r\;(A\,T\,G\,C)$}[.20\linewidth]{\includegraphics[width=\linewidth]{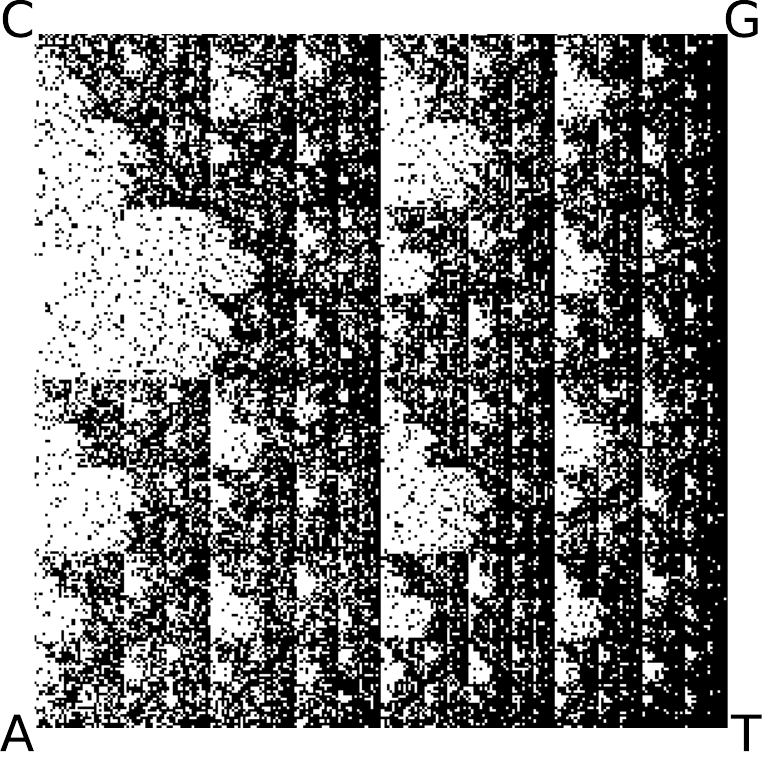}} &
    \subcaptionbox{\scriptsize$r^{2}\;(A\,G)(C\,T)$}[.20\linewidth]{\includegraphics[width=\linewidth]{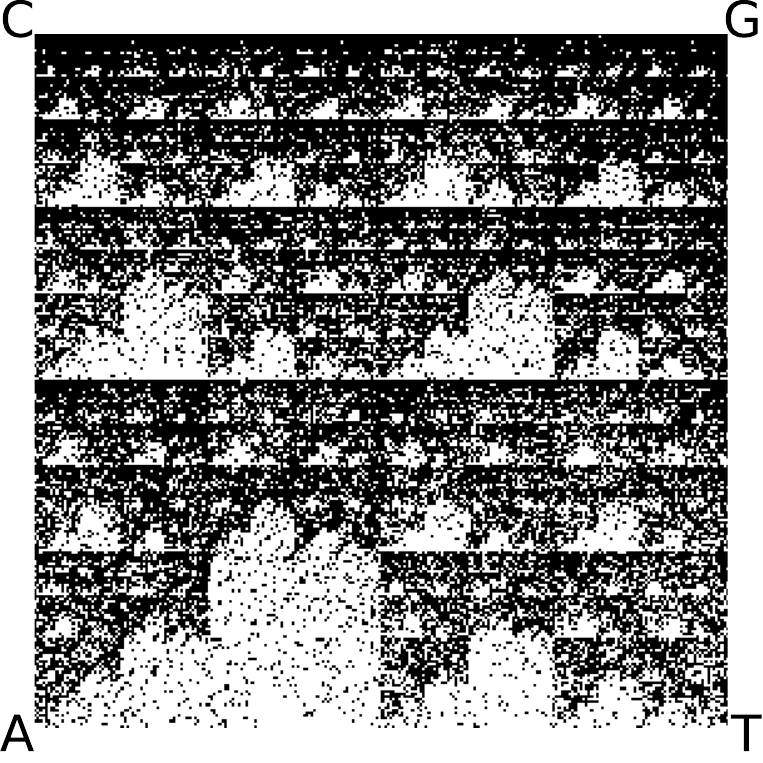}} &
    \subcaptionbox{\scriptsize$r^{3}\;(A\,C\,G\,T)$}[.20\linewidth]{\includegraphics[width=\linewidth]{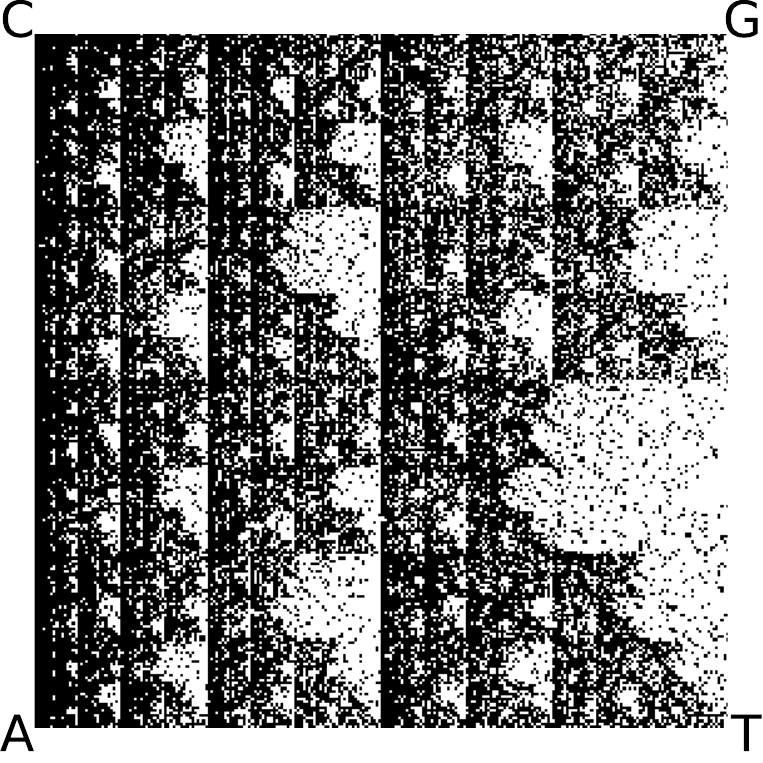}}\\
    \subcaptionbox{\scriptsize$s\;(A\,C)(G\,T)$}[.20\linewidth]{\includegraphics[width=\linewidth]{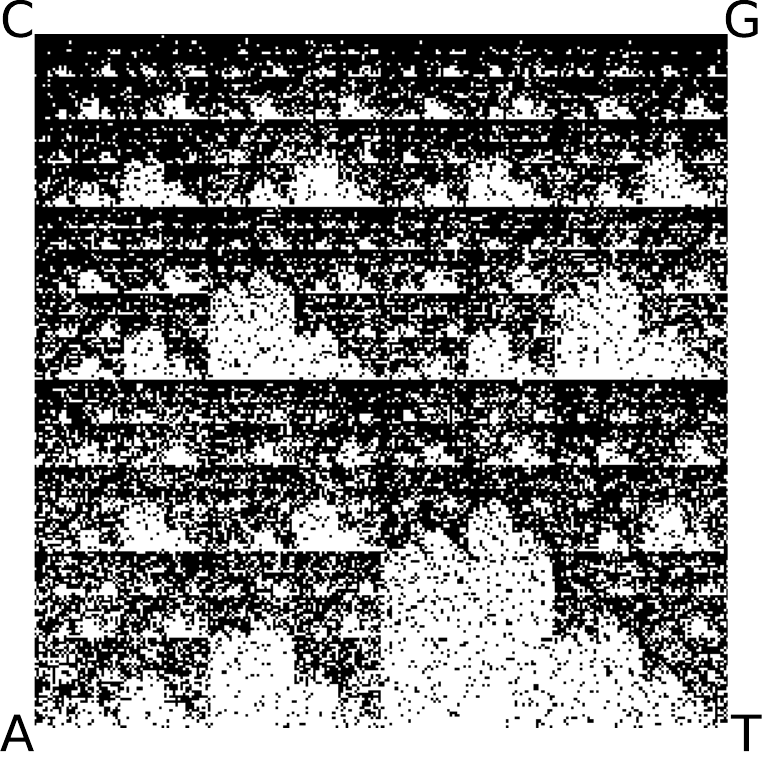}} &
    \subcaptionbox{\scriptsize$sr\;(A\,G)$}[.20\linewidth]{\includegraphics[width=\linewidth]{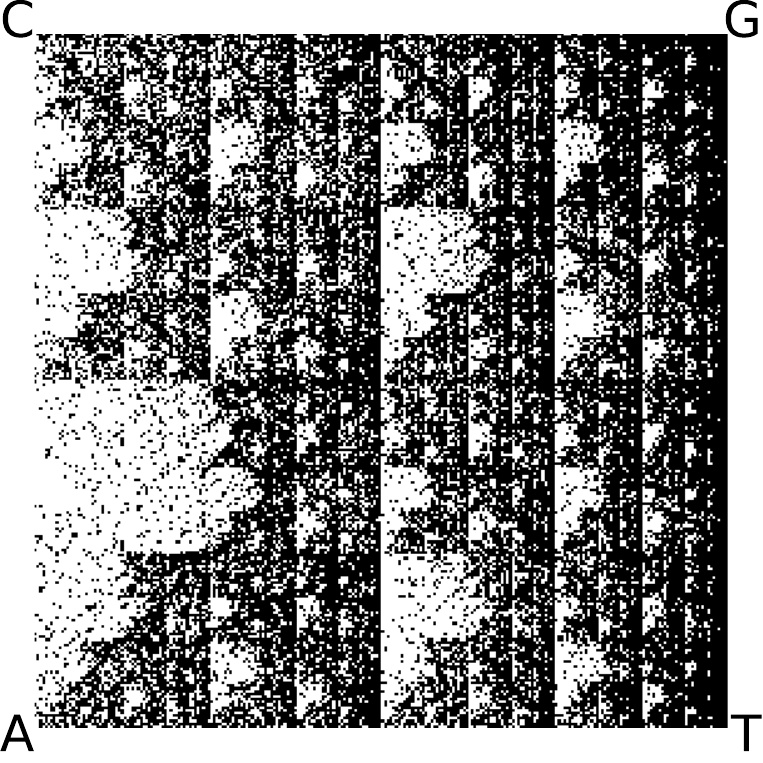}} &
    \subcaptionbox{\scriptsize$sr^{2}\;(A\,T)(C\,G)$}[.20\linewidth]{\includegraphics[width=\linewidth]{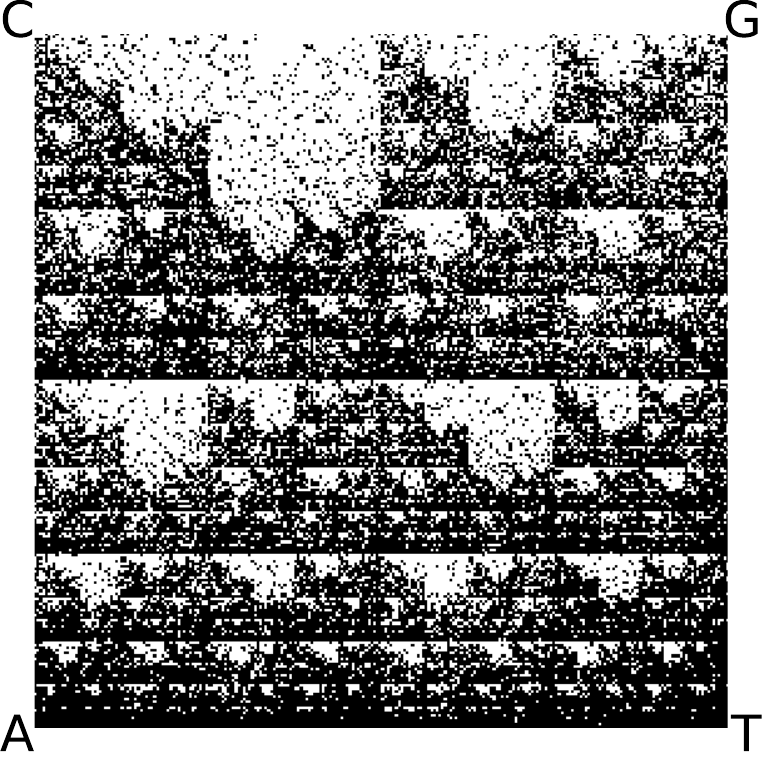}} &
    \subcaptionbox{\scriptsize$sr^{3}\;(C\,T)$}[.20\linewidth]{\includegraphics[width=\linewidth]{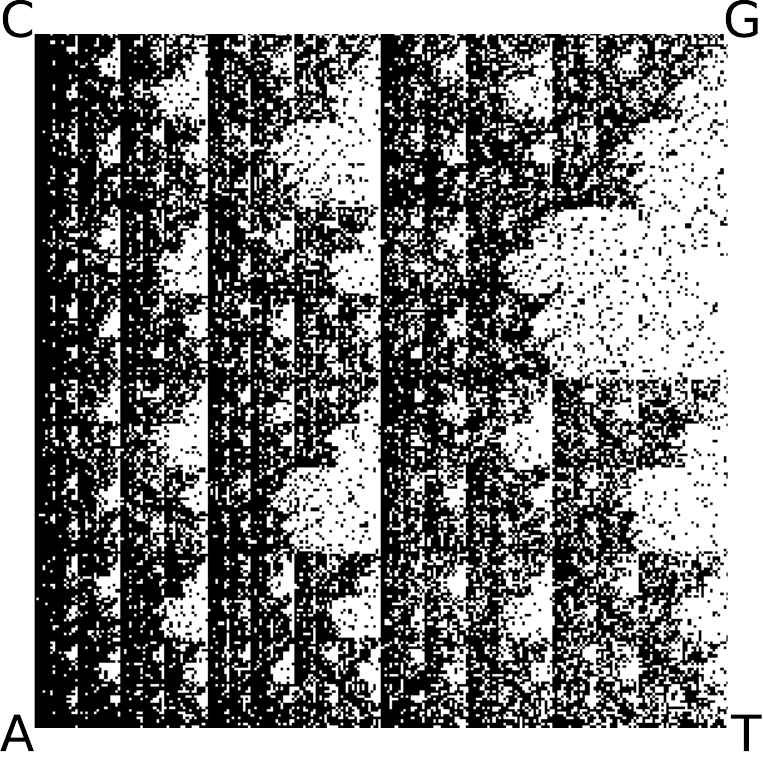}}
  \end{tabular}
  \caption{Visual correspondence between the eight elements of the permutation group $S_{4}$ acting on $\{A,C,G,T\}$ and the dihedral group $D_{8}$.  Panels (a) to (h) show Chaos Game Representations of a 100\,kb randomly selected fragment from human chromosome 4 (GRCh38.p14 primary assembly, RefSeq NC\_000004.12) after applying the indicated nucleotide permutation.}
  \label{fig:cgr-permutations}
\end{figure}

\section{FCGR via Kronecker Products}\label{subsec:fcgr_kron}
Kronecker product $\otimes$ of $A$ of size $m_1 \times n_1$ with $B$ of size $m_2 \times n_2$ is a matrix of size $m_1m_2 \times n_1n_2$ defined as follows:
 \[A \otimes B = \begin{bmatrix}
     A_{11} B & \cdots & A_{1n_1} B \\
     \vdots & \ddots & \vdots \\
      A_{m_{11}} B & \cdots & A_{m_1n_1} B
 \end{bmatrix}.\]

In \cite{hao2000}, FCGR is defined via Kronecker products of the $2 \times 2$ matrix as follows.
 \[ M = \begin{bmatrix}
     C & G \\
     A & T
 \end{bmatrix}, \quad  M^{(k)} = \underbrace{M \otimes M \otimes \cdots \otimes M}_{k \text{ times}}.\]

 \begin{definition}[\cite{hao2000}]\label{def:freq_def}
A $FCGR_k$ of a sequence $s \in \Sigma^n$ with order $k$ with $n \ge k$, is a matrix  $FCGR_k(s) \in \mathbb{R}^{2^k \times 2^k}$ such that its $(i,j)$th entry $FCGR_k(s)(i,j)$ is the frequency of the $k$-mer $M^{(k)}_{ij}$.
\end{definition}

 Note that Definition~\ref{def:freq_def} and Definition~\ref{def:grid_def} in Section~\ref{sec:fcgr} are not the same. As a counterexample, consider the case $k=2$. For a given $s \in \Sigma^n, n \geq k$, the FCGR entry $FCGR_k(s)(1,2)$ will be the frequency of $GC$ according to Definition \ref{def:grid_def} while according to Definition \ref{def:freq_def} from \cite{hao2000},  $FCGR_k(s)(1,2)$ will be the frequency of $CG$ instead. See Figure~\ref{fig:fcgr_comparison} for a visual comparison of FCGRs of order $k=7$ according to these two definitions for the 100,000 bp DNA sequence from human chromosome 4 (GRCh38.p14
primary assembly, RefSeq NC\_000004.12, used in Figure~\ref{fig:human_cgr}).

\begin{figure}[!pt]
    \centering
    \begin{subfigure}[b]{0.25\linewidth}
        \centering
        \includegraphics[width=\linewidth]{Figures/NC_000004.12.pdf}
         \vspace*{-5 pt}
        \caption{}
        \label{fig:human_cgr_comparison}
    \end{subfigure}
    \hspace{10 pt}
    \begin{subfigure}[b]{0.28\linewidth}
        \centering
        \includegraphics[width=\linewidth]{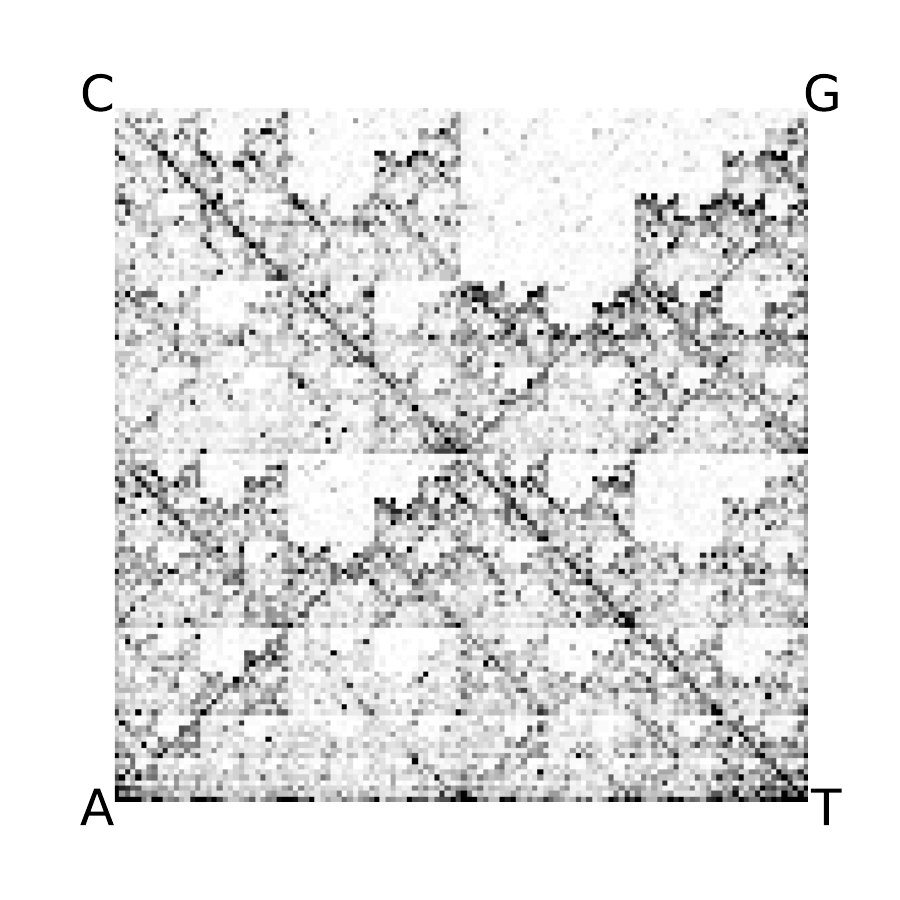}
        
        \caption{}
        \label{fig:fcgr_order_7}
    \end{subfigure}
    \hspace{10 pt}
    \begin{subfigure}[b]{0.28\linewidth}
        \centering
        \includegraphics[width=\linewidth]{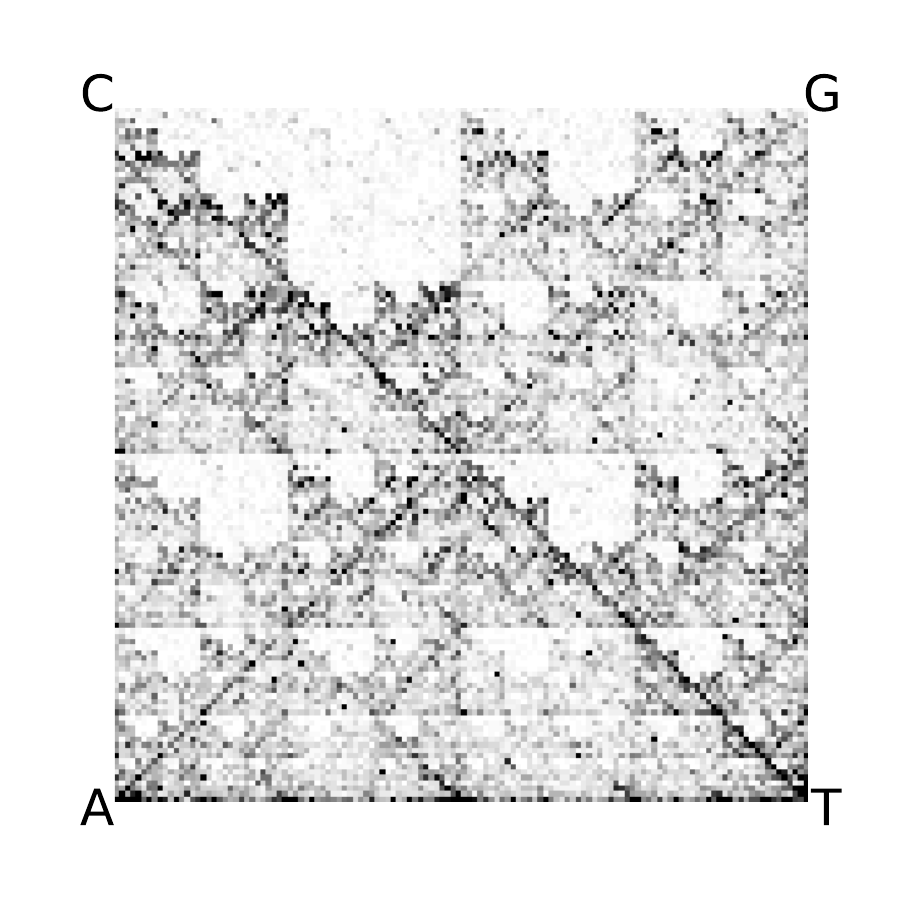}
        \caption{}
        \label{fig:fcgr_order7_hao2000}
    \end{subfigure}
    \caption{A visual comparison of FCGRs of order $k=7$ according to different definitions. The intensity of the pixel represents the frequencies of $7$-mers. Darker pixels mean larger frequencies.
    \textbf{(a)} CGR of a
100,000 bp DNA sequence randomly selected from human chromosome 4 (GRCh38.p14
primary assembly, RefSeq NC\_000004.12, used in Figure~\ref{fig:human_cgr}).
    \textbf{(b)} FCGR of order $7$ of the same sequence as in (a), according to our Definition~\ref{def:grid_def}.
    \textbf{(c)} FCGR of order $7$ of the same sequence as in (a), according to~\cite{hao2000}.
     }
    \label{fig:fcgr_comparison}
\vspace{-5mm}
\end{figure}

\section{De Bruijn graphs and their applications in sequence reconstruction}\label{sec:debruijn}
\begin{definition}[De Bruijn graph]
Let $k \geq 2$, $n > k$ and $s \in \Sigma^n$. The De Bruijn graph $\textsc{De Bruijn}_k(s)$ is the directed graph $(V_{k-1}(s), E_{k}(s))$, defined as follows: 
\begin{itemize}
\item  $V_{k-1}(s) = \{ v \in \Sigma^{k-1} \mid v \text{ appears as a substring of } s \}.$ 
\item $E_k(s) = \{ (u,v) \in V_{k-1}(s)^2 \mid \text{there exists a } k\text{-mer } w \in \Sigma^k \text{ that appears in } s \\
\text{ with } \textsc{Prefix}(w)=u \text{ and } \textsc{Suffix}(w)=v \}.$
\end{itemize}
\end{definition} 

An illustration of the graph \(\textsc{De  Bruijn}_2(\verb+ATCGTATCCA+)\) is presented in Figure~\ref{fig:debruijn_example}a). In this construction, each edge uniquely corresponds to a specific $k$-mer in the original string, so that every path through $\textsc{De Bruijn}_k(s)$ represents a substring of $s$ containing each $k$-mer exactly once. This characteristic, however, is insufficient for our purposes, as in our application, the goal is to reconstruct a sequence that preserves the target $k$-mer distribution; thus, it is necessary to employ a structure that accommodates multiple occurrences of the same $k$-mer. In other words, when a $k$-mer appears more than once in $s$, the standard De Bruijn graph is insufficient for recovering any substring corresponding to those repeated $k$-mers.

To address this issue, we introduce the notion of De Bruijn multigraph in Definition~\ref{def:db_multigraph}, a concept well established in the DNA sequence assembly literature. Note that an edge from vertex $u$ to vertex $v$ exists if the $k$-mer associated with that edge has  $u$ as its prefix and $v$ as its suffix, with the edge's multiplicity corresponding to the frequency of the associated $k$-mer in $s$. Figure~\ref{fig:debruijn_example}b) illustrates the multigraph $\textsc{De Bruijn}^{+}_2(\verb+ATCGTATCCA+)$.

\begin{figure}[!hpt]
    \centering
    \begin{minipage}[b]{0.45\linewidth}
        \centering
        \includegraphics[width=0.8\linewidth]{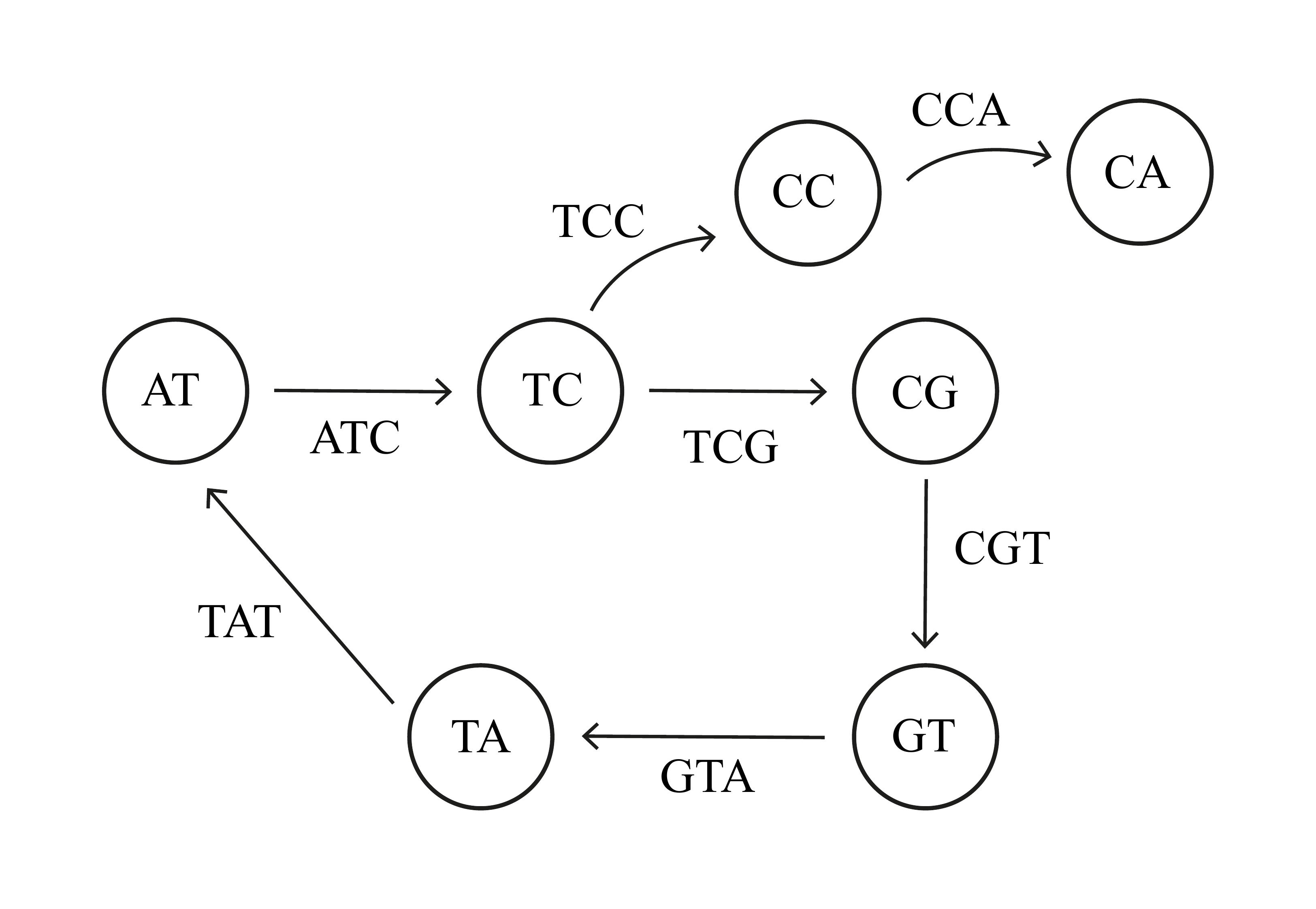}
        \caption*{(a) $\textsc{De Bruijn}_3(\texttt{ATCGTATCCA})$}
    \end{minipage}
    \hfill
    \begin{minipage}[b]{0.45\linewidth}
        \centering
        \includegraphics[width=0.8\linewidth]{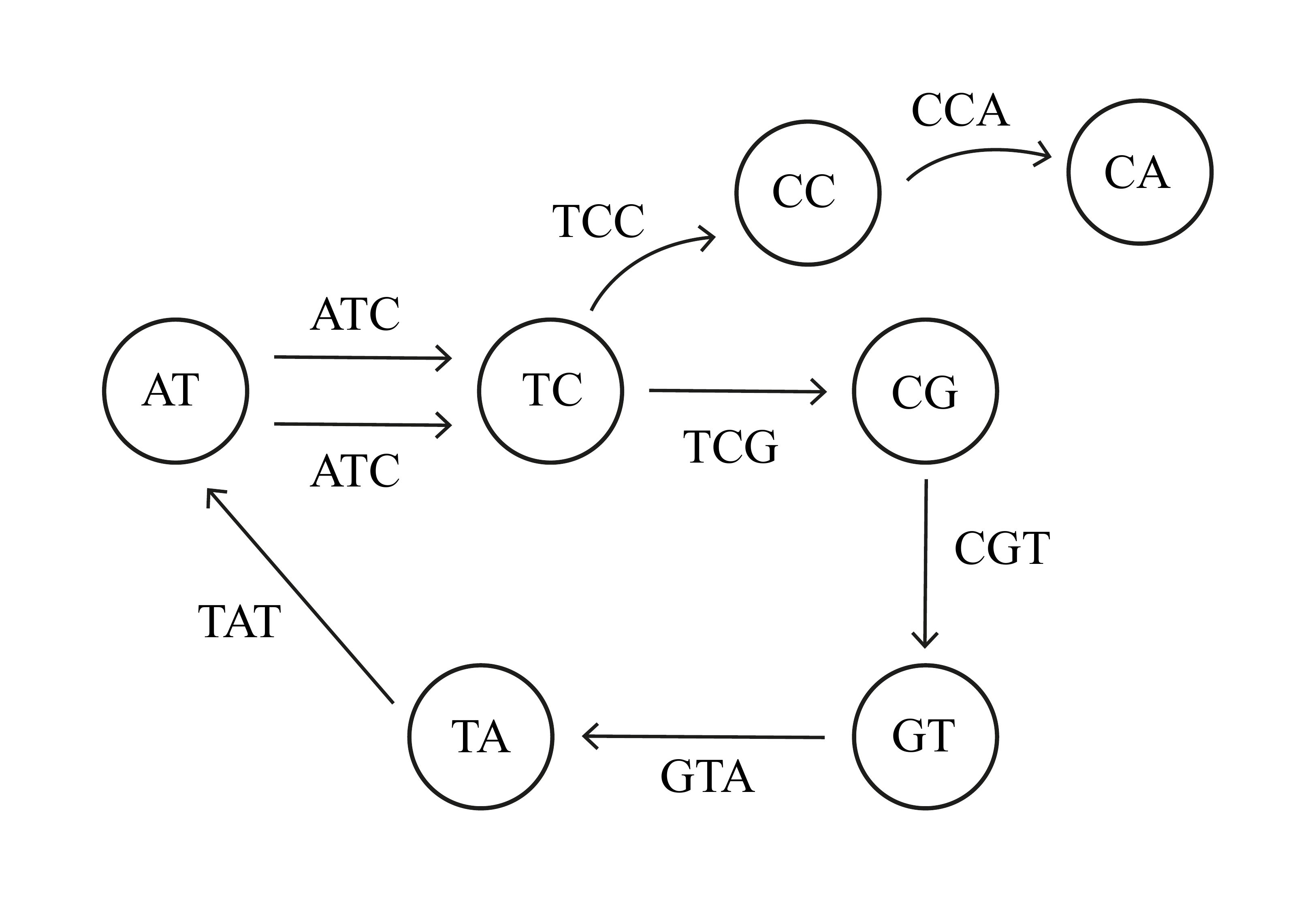}
        \caption*{(b) $\textsc{De Bruijn}^{+}_3(\texttt{ATCGTATCCA})$}
    \end{minipage}
    \caption{Comparison of De Bruijn graphs and multigraph of a sample string $s=\texttt{ATCGTATCCA}$, for $k=3$. Each vertex represents a unique 2-mer from the string, and each edge represents a unique 3-mer.  (a) De Bruijn graph: A walk in the graph corresponds to a substring of $s$ with each $k$-mer occurring exactly once. (b) De Bruijn multigraph: The edge labels indicate the multiplicity of each 3-mer, thereby allowing the recovery of any substring of $s$ via an appropriate Eulerian path.}
    \label{fig:debruijn_example}
\end{figure}

\section{Hit-and-run sampling from the standard probability simplex with marginalization constraints}
\label{appendix:hit-and-run}
Let $\vtheta_s = (\theta_w)_{w \in \Sigma^k}, \quad \Sigma = \{A,C,G,T\},$ denote the empirical $k$-mer frequency vector computed from a (linear) DNA sequence \(s\).  By construction,
\[\theta_w \ge 0, \quad \sum_{w\in\Sigma^k}\theta_w = 1.\]
In addition to this normalization, \(\vtheta_s\) satisfies a family of linear constraints induced by contiguous overlaps.  This is commonly called \emph{marginal consistency}: for each \(v\in\Sigma^{k-1}\),
\begin{equation}
\label{eq:marginal}
\sum_{a\in\Sigma} \theta_{v a}
= \sum_{a\in\Sigma} \theta_{a v} = \theta_v,
\end{equation}
where \(\theta_v\) is the frequency of the \((k-1)\)-mer \(v\). Combining these with the normalization yields a system
\begin{equation}
\label{eq:linear_constraints}
B\,\vtheta = \vb,
\end{equation}
where:
\begin{itemize}
  \item Row 0 of \(B\) is  \(\mathbf1^\top\), enforcing \(\sum_i x_i=1\).
  \item For each \((k-1)\)-mer \(v\), one row enforces
  \(\sum_{a\in\Sigma} x_{v a} - \sum_{a\in\Sigma} x_{a v} = 0\).

  \item $\vb = (1, 0,, \dots , 0)$
  
\end{itemize}

\begin{proposition}
\label{prop:full_rank}
The matrix \(B\) defined by \eqref{eq:linear_constraints} has full row rank.
\end{proposition}
\begin{proof}
Let the rows of \(B\) be \(r_0,r_{v_1},\dots,r_{v_M}\), where \(r_0=\mathbf1^\top\) and each \(r_v\) has nonzero entries only on indices corresponding to \(k\)-mers with prefix or suffix \(v\).  For distinct \(v\) and \(u\), their supports are disjoint, so \(r_{v}\) are mutually orthogonal and hence independent.  Moreover, \(r_0\) cannot be expressed as a linear combination of the zero‑sum rows \(r_v\).  Therefore no nontrivial linear relation among \(\{r_0,r_v\}\) exists, and \(B\) has full row rank.
\end{proof}

Let \(r=\mathrm{rank}(B)\).  The feasible set \(\{x\ge0:Bx=b\}\) is an \((4^k-r)\)-dimensional convex polytope.  We sample approximately uniformly from this polytope via a \emph{hit‑and‑run} Markov chain as follows:

\begin{enumerate}
  \item Compute an orthonormal basis \(N\in\mathbb R^{4^k\times(4^k-r)}\) of \(\ker(B)\) (e.g. via SVD).
  \item Initialize \(\vtheta^{(0)}=(1/4^k,\dots,1/4^k)\).
  \item For each iteration \(t=1,\dots,T\):
    \begin{enumerate}
      \item Draw \(\vz\sim\mathcal N(0,I_{4^k-r})\) and set \(\vd=N \vz/\|N\vz\|\).
      \item Find \(t_{\min}\le t_{\max}\) so that \(\vtheta^{(t-1)}+t\,\vd\ge0\) coordinate‑wise.
      \item Sample \(t\sim\mathrm{Unif}[t_{\min},t_{\max}]\).
      \item Update \(\vtheta^{(t)}=\vtheta^{(t-1)}+t\,\vd\).
    \end{enumerate}
  \item Return \(\vtheta^{(T)}\), which lies in the simplex and satisfies \(B\vtheta=b\).
\end{enumerate}

At each step, the direction \(d\) is drawn from a continuous distribution spanning all of \(\ker(B)\), so any two interior points can be connected by a finite sequence of positive‑probability moves. This procedure corresponds to a standard hit-and-run sampling strategy, common in MCMC theory (e.g.\ \cite{hit_and_run}).


\section{Real data acquisition and preprocessing} 
\label{appendix:real-dataset}
All genome data were obtained on 24 April 2025 from the NCBI RefSeq FTP server (\url{ftp.ncbi.nlm.nih.gov/genomes/refseq}). We considered eight taxonomic groups: $\texttt{Archaea}$, \texttt{Bacteria}, \texttt{Fungi}, \texttt{Plant}, \texttt{Protozoa}, \texttt{Invertebrate}, \texttt{Vertebrate\_Mammalian} and \texttt{Vertebrate\_Other}. The last three correspond to the kingdom \texttt{Animalia}. For each taxonomic group, the summary file was downloaded and parsed into a table of assembly metadata. Entries were grouped by the field \texttt{organism\_name}, and one assembly per species was selected by taking the first entry per group and then randomly sampling 100 species without replacement (random seed = 42). Assemblies lacking a valid FTP path or marked incomplete were omitted.\\

For each selected assembly, the corresponding FASTA file was downloaded via urllib (\texttt{Python 3.11}), decompressed with \texttt{gzip}, and parsed using \texttt{Biopython (v1.79)}. From each complete genome sequence, a single 100,000 bp fragment was extracted by choosing a start position uniformly at random from (seed = 42). Fragments shorter than 100 kbp were discarded. Each fragment was saved in FASTA format and CGR images (256×256 px) were generated in \texttt{NumPy (v1.23.5)} and \texttt{Matplotlib (v3.7.1)}.

All steps were automated by a \texttt{Python 3.9} script (available in Supplementary Software) to ensure reproducibility.

\section{Interactive synthetic CGR image generator}
\label{appendix:web-app}
To facilitate visual exploration of compositional biases in synthetic DNA, the research presented in this work is accompanied by an open‑source, browser‑based tool that generates sequences whose $k$-mer frequency vector matches a user‑specified or a sampled distribution.  The application supports rapid, visual inspection of how local changes in dinucleotide or higher‑order statistics alter global CGR patterns as the binary Chaos Game Representation (CGR) of the generated sequence is displayed to the user in each case. 

The user may choose a value of $k\in\{2,\dots,6\}$, a sequence length $L$ and may choose to generate the sequence from real DNA or to sample a point from the standard simplex subject to the marginal constraints. Additionally, for $k$=2, 16 log‑scaled sliders can be adjusted to select a target probability distribution (See Figure~\ref{fig:web-app}). The target frequency vector is converted to integer counts ($n-k+1$ observations) and embedded in a directed multigraph whose nodes are ($k-1$)-mers. A depth‑first Eulerian traversal yields a linear sequence that approximately produces the count multiset. The server returns a $256\times256$ binary CGR (PNG, base‑64 embedded). For $k\leq3$ it also returns a bar plot of the sampled probabilities.

The tool is implemented in \texttt{Python 3.11} using the following libraries: \texttt{Flask (v2.3)}, \texttt{NumPy (1.23.5)} and \texttt{Matplotlib (v3.8)}. The typical runtime is $<$ 1s for $n \leq 10^5$ and peak memory is less than 200MB on a standard laptop. The full source code can be found on GitHub under the MIT license. see \autoref{fig:web-app}.


%% file: UCNC.bbl
\begin{thebibliography}{10}
\providecommand{\url}[1]{\texttt{#1}}
\providecommand{\urlprefix}{URL }
\providecommand{\doi}[1]{https://doi.org/#1}

\bibitem{Alipour2024}
Alipour, F., Hill, K.A., Kari, L.: {CGRclust}: Chaos game representation for twin contrastive clustering of unlabelled {DNA} sequences. BMC Genomics  \textbf{25},  1--17 (2024)

\bibitem{AvilaCartes2023}
Avila~Cartes, J. et~al.: Accurate and fast clade assignment via deep learning and frequency chaos game representation. GigaScience  \textbf{12},  giac119 (2023)

\bibitem{SPADES_Pevzner}
Bankevich, A. et~al.: Spades: A new genome assembly algorithm and its applications to single-cell sequencing. J. Comput. Biol.  \textbf{19}(5),  455--477 (2012)

\bibitem{cgr-deschavanne}
Deschavanne, P.J. et~al.: {Genomic signature: characterization and classification of species assessed by chaos game representation of sequences.} Molecular Biology and Evolution  \textbf{16}(10),  1391--1399 (1999)

\bibitem{CART}
Dexter~Dyer, B., Kahnand, M.J., LeBlanc, M.D.: Classification and regression tree ({CART}) analyses of genomic signatures reveal sets of tetramers that discriminate temperature optima of archaea and bacteria. Archaea  \textbf{2},  159--167 (2008)

\bibitem{Abstractalgebra}
Dummit, D.S., Foote, R.M.: Abstract algebra. John Wiley \& Sons, Inc., Hoboken, NJ, third edn. (2004)

\bibitem{dutta1992}
Dutta, C., Das, J.: Mathematical characterization of chaos game representation: New algorithms for nucleotide sequence analysis. J. Molecular Biology  \textbf{228}(3),  715--719 (1992)

\bibitem{cgr_systematic_review_DeLaFuente}
De~la Fuente, R. et~al.: Genomic signature in evolutionary biology: A review. Biology  \textbf{12}(2) (2023)

\bibitem{goldman1993}
Goldman, N.: Nucleotide, dinucleotide and trinucleotide frequencies explain patterns observed in chaos game representations of {DNA} sequences. Nucleic Acids Research  \textbf{21}(10),  2487--2491 (1993)

\bibitem{hao2000}
Hao, B.L.: Fractals from genomes – exact solutions of a biology-inspired problem. Physica A: Statistical Mechanics and its Applications  \textbf{282}(1),  225--246 (2000)

\bibitem{Hill1992}
Hill, K.A., Schisler, N.J., Singh, S.M.: Chaos game representation of coding regions of human globin genes and alcohol dehydrogenase genes of phylogenetically divergent species. J. Molecular Evolution  \textbf{35},  261--269 (1992)

\bibitem{Hoang2016}
Hoang, T., Yin, C., Yau, S.S.T.: Numerical encoding of {DNA} sequences by chaos game representation with application in similarity comparison. Genomics  \textbf{108}(3--4),  134--142 (2016)

\bibitem{CGR}
Jeffrey, H.J.: {{Chaos game representation of gene structure}}. Nucleic Acids Research  \textbf{18}(8),  2163--2170 (1990)

\bibitem{Jeffrey1992}
Jeffrey, H.: Chaos game visualization of sequences. Computers \& Graphics  \textbf{16}(1),  25--33 (1992)

\bibitem{Joseph2006}
Joseph, J., Sasikumar, R.: Chaos game representation for comparison of whole genomes. BMC Bioinformatics  \textbf{7},  1--10 (2006)

\bibitem{cgr-kari-karamicahlis}
Karamichalis, R. et~al.: {{An investigation into inter-and intragenomic variations of graphic genomic signatures}}. BMC Bioinformatics  \textbf{16}(1) (2015)

\bibitem{signature_pops}
Karlin, S., Burge, C.: Dinucleotide relative abundance extremes: a genomic signature. Trends in Genetics  \textbf{11}(7),  283 -- 290 (1995)

\bibitem{Kislyuk2009}
Kislyuk, A. et~al.: Unsupervised statistical clustering of environmental shotgun sequences. BMC Bioinformatics  \textbf{10}(1), ~316 (2009)

\bibitem{Lichtblau2019}
Lichtblau, D.: Alignment-free genomic sequence comparison using {FCGR} and signal processing. BMC Bioinformatics  \textbf{20},  1--17 (2019)

\bibitem{LochelHeider2021}
L{\"o}chel, H.F., Heider, D.: Chaos game representation and its applications in bioinformatics. Comput. Struct. Biotechnol. J.  \textbf{19},  6263--6271 (2021)

\bibitem{hit_and_run}
Lov\'{a}sz, L., Vempala, S.: Hit-and-run from a corner. SIAM J. on Computing  \textbf{35}(4),  985--1005 (2006)

\bibitem{millan_arias_delucs_2022}
Mill{\'a}n~Arias, P.A. et~al.: {DeLUCS}: {Deep} learning for unsupervised clustering of {DNA} sequences. PLOS ONE  \textbf{17}(1),  e0261531 (2022)

\bibitem{MILLER2010315}
Miller, J.R., Koren, S., Sutton, G.: Assembly algorithms for next-generation sequencing data. Genomics  \textbf{95}(6),  315--327 (2010)

\bibitem{Nissen_taxonomic_clustering_VAMB}
Nissen, J.N. et~al.: {Improved metagenome binning and assembly using deep variational autoencoders}. Nature Biotechnology  \textbf{39}(5),  555--560 (2021)

\bibitem{oliver1993}
Oliver, J. et~al.: Entropic profiles of {DNA} sequences through chaos-game-derived images. J. Theoretical Biology  \textbf{160}(4),  457--470 (1993)

\bibitem{Pevzner_Eulerian_2001}
Pevzner, P.A., Tang, H., Waterman, M.S.: An {Eulerian} path approach to {DNA} fragment assembly. PNAS  \textbf{98}(17),  9748--9753 (2001)

\bibitem{Randhawa_taxonomic_classification_covid19}
Randhawa, G.S. et~al.: {Machine learning using intrinsic genomic signatures for rapid classification of novel pathogens: {COVID-19} case study}. PLOS ONE  \textbf{15}(4),  e0232391 (2020)

\bibitem{Rizzo2016}
Rizzo, R. et~al.: Classification experiments of {DNA} sequences by using a deep neural network and chaos game representation. In: Proc. of the 17th Intl. Conference on Computer Systems and Technologies. pp. 222--228 (2016)

\bibitem{Sengupta2020}
Sengupta, D.C. et~al.: Similarity studies of corona viruses through chaos game representation. Computational Molecular Bioscience  \textbf{10}(3), ~61 (2020)

\bibitem{Solis-Reyes-kameris}
Solis-Reyes, S. et~al.: {{An open-source k-mer based machine learning tool for fast and accurate sub-typing of HIV-1 genomes}}. PLOS ONE  \textbf{13}(11),  e0206409 (2018)

\bibitem{wang2005}
Wang, Y. et~al.: The spectrum of genomic signatures: from dinucleotides to chaos game representation. Gene  \textbf{346},  173--185 (2005)

\bibitem{AFBenchmark}
Zielezinski, A. et~al.: Benchmarking of alignment-free sequence comparison methods. Genome Biology  \textbf{20}(1), ~144 (2019)

\bibitem{Zielezinski_AF}
Zielezinski, A. et~al.: Alignment-free sequence comparison: benefits, applications, and tools. Genome Biology  \textbf{18}(1), ~186 (2017)

\end{thebibliography}
